\documentclass[11pt]{article} 

\usepackage[utf8]{inputenc} 

\usepackage[margin=1in,a4paper]{geometry} 
\usepackage{amsmath, amsthm, amssymb}

\usepackage{url}
\usepackage{comment}
\usepackage{todonotes}
\usepackage[noline,noend, boxed,linesnumbered]{algorithm2e}
\usepackage{enumitem}

\usepackage{authblk}
\usepackage{thmtools}
\usepackage{thm-restate}

\usepackage[pdftitle={The streaming k-mismatch problem},hidelinks]{hyperref}
\usepackage[capitalize]{cleveref}

\newcommand{\ourcomplexity}{$\Oh(\log\frac{n}{k}(\sqrt{k\log k} + \log^3{n}))$\xspace}
\newtheorem{thm}{Theorem}[section]

\newtheorem{fact}[thm]{Fact}
\newtheorem{conjecture}[thm]{Conjecture}
\newtheorem{cor}[thm]{Corollary}
\newtheorem{lemma}[thm]{Lemma}

\newtheorem{observation}[thm]{Observation}
\newtheorem{problem}[thm]{Problem}
\newtheorem{definition}[thm]{Definition}

\newcommand{\Oh}{\mathcal{O}}
\newcommand{\HD}{\operatorname{HD}}
\newcommand{\HS}{\operatorname{Ham}}
\newcommand{\Hs}[3]{\ensuremath{\operatorname{Ham}_{#1,#2}[#3]\xspace}} 
\newcommand{\MM}{\operatorname{MI}}
\newcommand{\NNN}{\mathbf{N}}
\newcommand{\CCC}{\mathbf{C}}
\newcommand{\KKK}{\mathbf{K}}

\newcommand{\Per}{\operatorname{Per}}
\newcommand{\rle}{\operatorname{rle}}
\newcommand{\PP}{\mathcal{P}}

\newcommand{\sm}{\setminus}
\newcommand{\sub}{\subseteq}
\newcommand{\C}{\mathcal{C}}
\newcommand{\floor}[1]{\left\lfloor #1 \right\rfloor}
\newcommand{\ceil}[1]{\left\lceil #1 \right\rceil}
\newcommand{\Fp}{\mathbb{F}_p}
\newcommand{\sk}{\operatorname{sk}}

\crefname{cor}{Corollary}{Corollaries}
\crefname{fact}{Fact}{Facts}

\usepackage[left,mathlines]{lineno}
\newcommand*\patchAmsMathEnvironmentForLineno[1]{%
  \expandafter\let\csname old#1\expandafter\endcsname\csname #1\endcsname
  \expandafter\let\csname oldend#1\expandafter\endcsname\csname end#1\endcsname
  \renewenvironment{#1}%
     {\linenomath\csname old#1\endcsname}%
     {\csname oldend#1\endcsname\endlinenomath}}%
\newcommand*\patchBothAmsMathEnvironmentsForLineno[1]{%
  \patchAmsMathEnvironmentForLineno{#1}%
  \patchAmsMathEnvironmentForLineno{#1*}}%
\AtBeginDocument{%
\patchBothAmsMathEnvironmentsForLineno{equation}%
\patchBothAmsMathEnvironmentsForLineno{align}%
\patchBothAmsMathEnvironmentsForLineno{flalign}%
\patchBothAmsMathEnvironmentsForLineno{alignat}%
\patchBothAmsMathEnvironmentsForLineno{gather}%
\patchBothAmsMathEnvironmentsForLineno{multline}%
}

\DeclareMathOperator*{\pl}{polylog}

\title{The streaming $k$-mismatch problem}
\author[1]{Rapha\"el Clifford}
\author[2]{Tomasz Kociumaka}
\author[3]{Ely Porat}

\affil[1]{Department of Computer Science, University of Bristol, United Kingdom}
\affil[ ]{\url{raphael.clifford@bristol.ac.uk}}
\affil[2]{Institute of Informatics, University of Warsaw, Poland}
\affil[ ]{\url{kociumaka@mimuw.edu.pl}}
\affil[3]{Department of Computer Science, Bar-Ilan University, Israel}
\affil[ ]{\url{porately@cs.biu.ac.il}}
\date{}

\begin{document}
\maketitle
\thispagestyle{empty}
\begin{abstract}
We consider the streaming complexity of a fundamental task in approximate pattern matching: the $k$-mismatch problem. It asks to compute Hamming distances between a pattern of length $n$ and all length-$n$ substrings of a text for which the Hamming distance does not exceed a given threshold $k$.  In our problem formulation, we report not only the Hamming distance but also, on demand, the full \emph{mismatch information}, that is the list of mismatched pairs of symbols and their indices.  The twin challenges of streaming pattern matching derive from the need both to achieve small working space and also to guarantee that every arriving input symbol is processed quickly.  

We present a streaming algorithm for the $k$-mismatch problem which uses $\Oh(k\log{n}\log\frac{n}{k})$ bits of space
and spends \ourcomplexity time on each symbol of the input stream, which consists of the pattern followed by the text.
The running time almost matches the classic offline solution~\cite{ALP:2004} and the space usage is within a logarithmic factor of optimal. 
 Our new algorithm therefore effectively resolves and also extends an open problem first posed in FOCS'09~\cite{Porat:09}. En route to this solution, we also give a deterministic $\Oh( k (\log \frac{n}{k} + \log |\Sigma|) )$-bit encoding of all the alignments with Hamming distance at most $k$ of a length-$n$ pattern within a text of length $\Oh(n)$. This secondary result provides an optimal solution to a natural communication complexity problem which may be of independent interest.
\end{abstract}
\setcounter{page}{0}
\newpage
\section{Introduction}

Combinatorial pattern matching has formed a cornerstone of both the theory and practice of algorithm design over a number of decades. Despite this long history, there has been a recent resurgence of interest in the complexity of the most basic problems in the field. This has been partly been fuelled by the discovery of multiple lower bounds conditioned on the hardness of a small set of well-known problems with naive solutions notoriously resistant to any significant improvement~\cite{AW:2014,AWW:2014,Bringmann:2014, BI:SETH:2015, ABW:2015,Bringmann:FOCS:2015,CGLS:2018}.  Pattern matching has also proved to be a rich ground for exploring the time and space complexity of streaming algorithms~\cite{Porat:09,CS:2010,EJS:2010,JPS:2013,BG:2014,CFPSS:2015,CS:2016,CFPSS:2016,EGSZ:2017} and it is this line of research that we follow.

We consider the most basic similarity measure between strings of different lengths: that of computing all Hamming distances between a pattern and equal-length substrings of a longer text.  
In the streaming $k$-mismatch problem, the input strings arrive one symbol at a time and the task is to output the Hamming distance between the pattern and the latest length-$n$ suffix of the text provided that it does not exceed a threshold $k$ specified in advance.

The problem of computing the exact Hamming distances between a pattern and every length-$n$ substring of a text of length $\Oh(n)$ has been studied in the standard offline model for over 30 years. In 1987,  $\Oh(n\sqrt{n\log{n}})$-time solutions were first developed~\cite{DBLP:journals/siamcomp/Abrahamson87,Kosaraju:1987}.  Motivated by the need to find close matches quickly, from there the focus moved to the bounded $k$-mismatch version of the problem. For many years, the fastest solution ran in $\Oh(nk)$ time~\cite{LV:1986a}.  It was not until 2000, when a breakthrough result gave $\Oh(n\sqrt{k\log{k}})$ time~\cite{ALP:2004}. Much later, an $\Oh(k^2 \log k + n\pl n)$-time solution was developed~\cite{CFPSS:2016},
and a recent manuscript~\cite{GU:2017} improves this further to $\Oh((n+k\sqrt{n})\pl n)$ time.

Considered as an online or streaming problem with one text symbol arriving at a time, $k$-mismatch admits a linear-space solution running in $\Oh(\sqrt{k}\log{k} + \log{n})$ worst-case time per arriving symbol, as shown in 2010~\cite{CS:2010}.  The  paper of Porat and Porat in FOCS'09 gave an $\Oh(k^3\pl n)$-space and $\Oh(k^2 \pl n)$-time streaming solution, showing for the first time that the $k$-mismatch problem could be solved in sublinear space for particular ranges of $k$~\cite{Porat:09}.   In SODA'16, this was subsequently improved to $\Oh(k^2\pl n)$ space and $\Oh(\sqrt{k}\log{k}+\pl  n)$ time per arriving  symbol~\cite{CFPSS:2016}.  As we will describe below, our solution tackles a harder version of the $k$-mismatch problem while having nearly optimal space complexity $\Oh(k \pl n)$.

The twin challenges of streaming pattern matching stem from the need to optimise both working space and a guarantee on the running time for every arriving symbol of the text.  
An important feature of the streaming model is that we must account for all the space used and cannot, for example, store a copy of the pattern. 

One can derive a space lower bound for any streaming problem by looking at a related one-way communication complexity problem. The randomised one-way communication complexity of determining if the Hamming distance between two strings is greater than $k$ is known to be $\Omega(k)$ bits with an upper bound of $\Oh(k\log{k})$ bits~\cite{HSZZ:06}.  In our problem formulation, however, we report not only the Hamming distance but also the full \emph{mismatch information}---the list of mismatched pairs of symbols and their indices.  In this situation, one can derive a slightly higher space lower bound of $\Omega(k(\log{\frac{n}{k}} + \log{|\Sigma|}))$ bits\footnote{This follows directly from  the observation that for a single alignment with Hamming distance $k$, there are $\binom{n}{k}$  possible sets of  mismatch  indices and each of the $k$ mismatched symbols requires $\Omega(\log |\Sigma|)$  bits to be represented, where $\Sigma$ denotes  the input alphabet.  From this, we derive the same lower bound for the space required by any streaming $k$-mismatch algorithm. We assume throughout that $|\Sigma|$ is bounded by a polynomial in  $n$.}.  This formulation has been tackled before as a streaming pattern matching problem: an $\Oh(k^2 \log^{10}n /\log{\log{n}})$ space and $\Oh(k\log^{8}{n}/\log{\log{n}})$ time solution was given~\cite{RS:2017}.  Prior to the work we present here, the simple lower bound for a single output combined with the upper bounds presented above used to represent the limits of our understanding of the complexity of this basic problem.   

In this paper, we almost completely resolve both the time and space complexity of the streaming $k$-mismatch pattern matching problem.  Our solution is also the first small-space streaming pattern matching algorithm for anything other than exact matching which requires no offline and potentially expensive preprocessing of the pattern. 
That is, it assumes that the pattern precedes the text in the input stream and processes it in a truly streaming fashion.

\begin{problem}\label{prob:streaming}
	Consider a pattern of length $n$ and a longer text which arrive in a stream one symbol at a time.  The streaming $k$-mismatch problem asks after each arriving symbol 
	of the text whether the current suffix of the text has Hamming distance at most $k$ with the pattern and if so,  it also asks to return the corresponding mismatch  information. 
\end{problem}

Our main result is an algorithm for the streaming $k$-mismatch problem which almost matches the running time of the classic offline algorithm while using nearly optimal working space.  Unlike the fastest previous solutions for streaming $k$-mismatch (see e.g.\@~\cite{Porat:09,CFPSS:2016}), the algorithm we describe also allows us to report the full set of mismatched symbols  (and their indices) at each $k$-mismatch alignment. This gives a remarkable resolution to the complexity of the  streaming $k$-mismatch problem first posed by Porat and Porat in FOCS'09~\cite{Porat:09}.


\begin{restatable}{theorem}{thmstreaming}\label{thm:streaming}
	There exists a streaming $k$-mismatch algorithm which uses $\Oh(k\log{n}\log\frac{n}{k})$ bits of space and takes \ourcomplexity  time per arriving symbol.  The algorithm is randomised and its answers are correct with high probability, that is it errs with probability inverse polynomial~in~$n$. For each reported occurrence, the mismatch information can be reported on demand in $\Oh(k)$ time.
\end{restatable}
While processing the pattern, our streaming algorithm works under the same restrictions on space consumption and per-symbol running time as when processing the text.
This is in contrast to previous work on streaming approximate pattern matching which has typically included a potentially time- and space-inefficient offline preprocessing stage.  

In order to achieve our time and space improvements, we develop a number of new ideas and techniques which we believe may have applications more broadly. The first is a randomised $\Oh(k\log{n})$-bit sketch which allows us not only to detect if two strings of the same length have Hamming distance at most $k$ but if they do, also to report the related mismatch information.  The sketch we give has a number of desirable algorithmic properties, including the ability to be efficiently maintained subject to concatenation and prefix removal.

Armed with such a rolling sketch, one approach to the $k$-mismatch streaming problem could simply be to maintain the sketch of the length-$n$ suffix of the text and to compare it to the sketch of the whole pattern.  Although this takes $\Oh(k\pl n)$ time per arriving symbol, it would also require $\Oh(n\log |\Sigma|)$ bits of space to retrieve the leftmost symbol that has to be removed from the sliding window at each new alignment.  
This is the central obstacle in streaming pattern matching which has to be overcome.


Following the model of previous work on streaming exact pattern matching (see~\cite{Porat:09,BG:2014}), we introduce a family of $\Oh(\log n)$ prefixes $P_\ell$ of the pattern with exponentially increasing lengths. We organise the algorithm into several \emph{levels}, with the $\ell$th level responsible for finding the $k$-mismatch occurrences of $P_\ell$.
The task of the next level is therefore to check,
 after $|P_{\ell+1}|-|P_{\ell}|$ subsequent symbols are read, which of these occurrences can be extended to $k$-mismatch occurrences of $P_{\ell+1}$.
The key challenge is that in the meantime the $k$-mismatch occurrences of $P_{\ell}$ have to be stored in a space-efficient way.
In the exact setting, their starting positions form an arithmetic progression,
but the presence of mismatches leads to a highly irregular structure.

The task of storing $k$-mismatch occurrences of a pattern in a space-efficient representation
can be expressed in terms of a natural communication problem which is of independent interest.
 
\begin{problem}\label{prob:one}
Alice has a pattern $P$ of length $n$ and a text $T$ of length $\Oh(n)$. She sends one message to Bob, who holds neither the pattern nor text. 
Bob must output all the alignments of the pattern and the text with at most $k$ mismatches, as well as the applicable mismatch information. 
\end{problem}

The solution we give for this problem is both deterministic and asymptotically optimal.

\begin{restatable}{theorem}{thmone}\label{thm:prob-one}
There  exists a deterministic one-way communication complexity protocol for Problem~\ref{prob:one} that sends $\Oh(k(\log{\frac{n}{k}}+ \log |\Sigma|))$ bits,
where $\Sigma$ denotes the input alphabet.
\end{restatable}

One striking property of this result is that the communication complexity upper bound matches the lower bound we gave earlier for two strings of exactly the same length.  
In other words, we require no more space to report the mismatch information at all $k$-mismatch alignments than we do to report it at only one such alignment. 

As the main conceptual step in our solution to \cref{prob:one}, we introduce modified versions of the pattern and text which are highly compressible but still contain sufficient information to solve the $k$-mismatch problem. More specifically, we place sentinel symbols at some positions of the text and the pattern,
making sure that no new $k$-mismatch occurrences are introduced and that the mismatch information for the existing occurrences does not change.
We then develop a key data structure (specified in \cref{lem:key}) which lets us store the proxy pattern in a space-efficient way
with $\Oh(\log n)$-time random access to any symbol.
On the other hand, the relevant part of the text is covered by a constant number of $k$-mismatch occurrences of the pattern,
so the modified text can be retrieved based on the proxy pattern and the mismatch information for these $\Oh(1)$ occurrences. 
Bob may use this data to find all the $k$-mismatch occurrences of the pattern in the text along with their mismatch information. 

We go on to show that both the encoding and decoding steps can be implemented quickly and in small space. 
The key tool behind the efficient decoding is a new fast and space-efficient algorithm for the $k$-mismatch problem in a setting with read-only random access to the input strings.
This is much easier compared to the streaming model as the aforementioned idea of maintaining the sketch of a sliding window now requires just $\Oh(k\log n)$ bits. This simple method would, however, be too slow for our purposes, and so in \cref{thm:ram} we give a more efficient solution to the $k$-mismatch problem in the read-only model. Then in \cref{thm:buffer} we use this to give a solution for \cref{prob:one} which is not only time- and space-efficient but which also lets us store the $k$-mismatch occurrences of prefixes of the pattern
and retrieve them later on when they are going to be necessary. In our main algorithm, we apply it to store the $k$-mismatch occurrences of $P_{\ell}$ until we can try extending them to $k$-mismatch occurrences of $P_{\ell+1}$.

In order to guarantee that we can always process every symbol of the text in $\Oh(\sqrt{k} \pl n)$ time rather than in $\Oh(k\pl n)$ time, we develop a different procedure for matching strings with small approximate periods.  The difficulty with such strings is that their $k$-mismatch occurrences may occur very frequently.
Given in \cref{thm:small-period}, our solution is based on a novel adaptation of Abrahamson's algorithm~\cite{DBLP:journals/siamcomp/Abrahamson87} designed for space-efficient convolution of sparse vectors.
We apply it at the lowest level of our streaming algorithm so that at the higher levels we can guarantee that the $k$-mismatch occurrences of $P_\ell$
start at least $k$ positions apart. 

Finally, the use of sketches incurs a delay of $\Oh(k \pl n)$ time when a $k$-mismatch occurrence is verified.
To achieve better worst-case time per symbol, we set the penultimate prefix $P_{L-1}$ to be of length $n-2k$ and we make the last level naively check
which $k$-mismatch occurrences of $P_{L-1}$ extend to $k$-mismatch occurrences of $P=P_L$. This trick
lets us start verification already while reading the trailing $2k$ symbols of the candidate length-$n$ substring of the text.

In summary, our main contribution is given by \cref{thm:streaming}, but in order to achieve this, we have developed a number of new tools and techniques. These include a new sliding window sketch, a new small approximate period algorithm in \cref{thm:small-period},  the communication protocol of \cref{thm:prob-one}, a new read-only pattern matching algorithm in \cref{thm:ram} and, most importantly, our key technical innovation given by \cref{thm:buffer}.
This last result demonstrates that despite few structural properties,  overlapping $k$-mismatch occurrences admit a very space-efficient representation
with a convenient algorithmic interface.

\section{A Rolling $k$-mismatch Sketch}\label{sec:fingerprint}
In this section, we give an overview of our new rolling sketch which will not only allow us to determine if two strings have Hamming distance at most $k$ but if they do, it will also give us all the mismatch information.  
Our approach extends the deterministic sketch developed in~\cite{CEPR:2009} for the offline $k$-mismatch with wildcards problem and combines it with the classic Karp--Rabin fingerprints for exact matching~\cite{DBLP:journals/ibmrd/KarpR87}.

We fix an upper bound $n$ on the length of the compared strings and a prime number $p > n^c$ for sufficiently large exponent $c$ (the parameter $c$ can be used to control error probability). 
We will assume throughout that all the input symbols can be treated as elements of $\Fp$ by simply reading the bit representation of the symbols.
If the symbols come from a larger alphabet, then we would need to hash them into $\Fp$, which will introduce a small extra probability of error. 

Note that the earlier sketch of~\cite{CEPR:2009} is based on fields with characteristic two.
However, in order to make our sketch able to roll forwards, we need to perform computations in a field with large characteristic (larger than $n$).  
The downside of this change is that we have to use a randomised polynomial factorisation algorithm to find the indices of the mismatches.

Let us start by recalling the Karp--Rabin fingerprints~\cite{DBLP:journals/ibmrd/KarpR87} and defining our new sketch.
\begin{fact}[Karp--Rabin fingerprints]
	For $r\in \Fp$ chosen uniformly at random, the Karp--Rabin fingerprints $\psi_r$, defined as 
	$\psi_r(S)=\sum_{i=0}^{\ell-1} S[i]r^{i}$ for $S\in \Fp^\ell$, 
	satisfy the following property:	if $U,V\in \Fp^\ell$ are not equal, then $\psi_r(U)=\psi_r(V)$ holds with probability at most $\frac{\ell}{p}$.
\end{fact}

\begin{definition}[$k$-mismatch sketch]
	For a fixed prime number $p$ and for $r\in \Fp$ chosen uniformly at random,
	the sketch $\sk_k(S)$ of a string $S\in \Fp^*$ is defined as:
		\[\sk_k(S) = (\phi_0(S),\ldots,\phi_{2k}(S), \phi'_0(S),\ldots,\phi'_k(S), \psi_r(S)),\]
	where $\phi_j(S)=\sum_{i=0}^{\ell-1}S[i]i^j$ and $\phi'_j(S)=\sum_{i=0}^{\ell-1}S[i]^2 i^j$ for $j\ge 0$.
\end{definition}

Observe that the sketch is a sequence of $3k+3$ elements of $\Fp$,
so it takes $\Oh(k\log p)=\Oh(k\log n)$ bits.
The main goal of the sketches is to check whether two given strings are at Hamming distance $k$ or less,
and, if so, to retrieve the mismatches. We define the mismatch information between two strings $X$ and $Y$ as $\MM(X,Y) = \{(i, X[i], Y[i]) : X[i] \ne Y[i]\}$.

\begin{restatable}{lemma}{corsketches}\label{cor:sketches}
	Given the sketches $\sk_k(S)$ and $\sk_k(T)$ of two strings of the same length $\ell \le n$, in $\Oh(k\log^3 n)$ time we can decide (with high probability)
	whether $\HD(S,T) \le k$. If so, the mismatch information $\MM(S,T)$ is reported. The algorithm uses $\Oh(k\log n)$ bits of space.
\end{restatable}

Next, we consider the efficiency of updating a sketch given the mismatch information.
\begin{restatable}{lemma}{lemconstr}\label{lem:constr}
Let $S,T\in \Fp^*$ be of the same length $\ell < n$.
If $\HD(S,T)= \Oh(k)$,  then $\sk_k(T)$ can be constructed in $\Oh(k\log^2 n)$ time and $\Oh(k\log n)$ bits of space 
given $\MM(S,T)$ and  $\sk_k(S)$.
\end{restatable}
As a result, the sketch can be efficiently maintained subject to elementary operations.
\begin{restatable}{cor}{corstreamsketch}\label{cor:streamsketch}
	A string $X\in \Fp^*$ with $|X|\le n$ can be stored in $\Oh(k\log n)$ bits 
so that $\sk_k(X)$ can be retrieved in $\Oh(k\log^2 n)$ time and the following updates
are handled in $\Oh(\log^2 n)$ time:
	\begin{enumerate}[itemsep=0pt,topsep=5pt,parsep=2pt]
	  \item append a given symbol $a\in \Fp$ to $X$, 
	  \item substitute $X[i]=a$ for $X[i]=b$ given the index $i$ and the symbols $a,b\in \Fp$.
	\end{enumerate}
\end{restatable}
Unlike in~\cite{CEPR:2009}, we also need to efficiently maintain sketches subject to concatenation etc.
\begin{restatable}{lemma}{lemmanip}\label{lem:manip}
The following operations can be implemented in $\Oh(k \log n)$ time using $\Oh(k\log n)$ bits of space,
provided that all the processed strings belong to $\Fp^*$ and are of length at most $n$.
\begin{enumerate}[itemsep=0pt,topsep=5pt,parsep=2pt]
  \item Construct one of the sketches  $\sk_k(U)$, $\sk_k(V)$, or $\sk_k(UV)$ given the other two. 
  \item Construct $\sk_k(U)$ or $\sk_k(U^m)$ given the other sketch and the integer $m$.
\end{enumerate}
\end{restatable}

\section{Patterns with a Small Approximate Period}\label{sec:periodic}
As our first space- and time-efficient algorithm, 
we show how the streaming $k$-mismatch problem can be solved deterministically when the pattern $P$ 
has a small approximate period and therefore can be stored in $\Oh(k)$ words of space.

Recall that a string $X$ of length $n$ is defined to have a period $p>0$ if ${X[0,\ldots,n - p - 1]}={X[p,\ldots,n-1]}$. 
The $d$-periods describe analogous structure for the setting with mismatches:
We say that an  integer $p$ is a \emph{$d$-period} of a string $X$ of length $n$ if $\HD({X[0,\ldots,n-p-1]},$ ${X[p,\ldots,n-1]})\le d$. 
The set of $d$-periods of $X$ is denoted by $\Per(X, d)$.

A string $X$ with a $d$-period $p$ can be stored using an $\Oh(d\log\frac{n}{d}+(d+p)\log |\Sigma|)$-bit
\emph{periodic representation with respect to $p$}, which by definition consists
of $X[0,\ldots,p-1]$ and $\MM(X[0,\ldots,n-p-1],X[p,\ldots,n-1])$; see~\cite{CFPSS:2016}.
The following lemma lets us detect a small $d$-period of a given string $X$
and construct the underlying periodic representation.
If $X$ does not have any such $d$-period, we can still retrieve the longest prefix of $X$ which has one.
This feature is going to be useful in \cref{sec:main-algorithm}, where we solve the general $k$-mismatch problem.

\begin{restatable}{lemma}{lempp}\label{lem:perpref}
There exists a deterministic streaming algorithm that, given positive integers $p$ and $d=\Oh(p)$,
finds the longest prefix $Y$ of the input string $X$ which has a $d$-period $p'\le p$.
It reports the periodic representation of $Y$ with respect to $p'$,
uses $\Oh(p)$ words of space, and takes $\Oh(\sqrt{p\log p})$ per-symbol processing time plus $\Oh(p\sqrt{p\log p})$ post-processing time.
\end{restatable}

In~\cite{CFPSS:2016}, it was proved that any $k^2$ consecutive values 
$\HS_{P,T}[i] := \HD(P, T[i-|P|+1,\ldots,i])$
can be generated in $\Oh(k^2\log k)$ time using $\Oh(k^2)$ words of space if the pattern $P$ and the text $T$
share a $d$-period $p$ satisfying $d=\Oh(k)$ and $p=\Oh(k)$. 
Below, we show how to compute $k$ subsequent Hamming distances in $\Oh(k\sqrt{k\log k})$ time 
using $\Oh(k)$ words of space, under an additional assumption that the preceding $2p$ Hamming distances
are already available.

Our approach resembles Abrahamson's algorithm~\cite{DBLP:journals/siamcomp/Abrahamson87},
so let us first recall how the values $\HS_{P,T}[i]$ can be expressed in terms of convolutions.
The \emph{convolution} of two functions $f,g: \mathbb{Z}\to \mathbb{Z}$ is a function
$f * g : \mathbb{Z}\to \mathbb{Z}$ such that
\[(f* g)(i) = \sum_{j\in \mathbb{Z}}f(j)\cdot g(i-j).\]

For a string $X$ and a symbol $a\in \Sigma$, we define a \emph{characteristic function} $X_a : \mathbb{Z}\to \{0,1\}$
of positions where $a$ occurs in $X$.
In other words, $X_a[i]=1$ if and only if $0\le i < |X|$ and $X[i]=a$.
The \emph{cross-correlation} of strings $T$ and $P$ is a function $T\otimes P  : \mathbb{Z}\to \mathbb{Z}$
defined as $T \otimes P = \sum_{a \in \Sigma} T_a * P^R_a$, where $P^R$ denotes the reverse of $P$.

\begin{restatable}{fact}{fctotimesham}\label{fct:otimesham}
We have $(T\otimes P)(i) = |P|-\HS_{P,T}[i]$ for $|P|-1 \le i  < |T|$
and $(T\otimes P)(i)=0$ for $i< 0$ and $i \ge |P|+|T|$.
\end{restatable}

If a string $X$ has a $d$-period $p$, then $X_a[i]$ is typically equal to $X_a[i+p]$.
This property can be conveniently formalised using a notion of \emph{finite differences}.
For a function $f:\mathbb{Z}\to \mathbb{Z}$ and a positive integer $p\in \mathbb{Z}_{+}$,
we define the \emph{forward difference} $\Delta_p[f] : \mathbb{Z}\to \mathbb{Z}$ 
as 
\[\Delta_p[f](i) = f(i+p)-f(i).\]
\vspace{-.5cm}
\begin{observation}
If a string $X$ has a $d$-period $p$,
then the functions $\Delta_p[X_a]$ have at most $2(d+p)$ non-zero entries in total across all $a\in \Sigma$.
\end{observation}

The following lemma reuses the idea behind Abrahamson's algorithm
to compute the convolution of functions with a sparse support, i.e., with few non-zero entries.
\begin{restatable}{lemma}{lemfft}\label{lem:fft}
Consider functions $f,g : \mathbb{Z}\to \mathbb{Z}$ with at most $n$ non-zero entries in total.
The non-zero entries among any $\delta$ consecutive values $(f*g)(i),\ldots,(f*g)(i+\delta-1)$ can be computed in
$\Oh(n\sqrt{\delta\log \delta})$ time using $\Oh(n+\delta)$ words of working space.
\end{restatable}
This is very useful because the forward difference operator commutes with the convolution:%
\begin{restatable}{fact}{fctdiff}\label{fct:diff}
	Consider functions  $f,g : \mathbb{Z}\to \mathbb{Z}$ with finite support and a positive integer $p$. We have
	$\Delta_p[f * g] = f * \Delta_p[g] = \Delta_p[g] * f$. Consequently, $\Delta_p[f]*\Delta_p[g]=\Delta_p[\Delta_p[f * g]]$.
\end{restatable}
The function $\Delta_p[\Delta_p[h]]$, called the \emph{second forward difference} of $h:\mathbb{Z}\to \mathbb{Z}$,
is denoted $\Delta_p^2[h]$; observe that $\Delta^2_p[h](i) = h(i+2p)-2h(i+p)+h(i)$.

Combining \cref{lem:fft}, \cref{fct:diff}, and the notions introduced above, we can compute the second forward differences of the cross-correlation between $P$ and $T$ efficiently  and in small space:
\begin{cor}\label{cor:mpt}
Suppose that $p$ is a $d$-period of strings $P$ and $T$.
Given the periodic representations of $P$ and $T$ with respect to $p$,
any $\delta$ consecutive values $\Delta_p^2[T\otimes P](i), \ldots, \Delta_p^2[T\otimes P](i+\delta-1)$
can be computed in $\Oh(\delta+(d+p)\sqrt{\delta\log \delta})$ time using $\Oh(d+p+\delta)$ words of space.
\end{cor}
\begin{proof}
The functions $\Delta_p[P^R_a]$ have $2(d+p)$ non-zero entries in total,
and the functions $\Delta_p[T_a]$ enjoy the same property. 
Hence, using \cref{lem:fft} to compute all the non-zero entries among
$(\Delta_p[T_a]*\Delta_p[P^R_a])(j)$ for $a\in \Sigma$ and $i \le j < i+\delta$
takes $\Oh((d+p)\sqrt{\delta\log \delta})$ time in total.
Finally, we observe that $\Delta_p^2[T\otimes P] = \sum_{a\in \Sigma} (\Delta_p[T_a]*\Delta_p[P^R_a])$
by \cref{fct:diff}.
\end{proof}

\Cref{fct:otimesham,cor:mpt} can be applied to compute the subsequent Hamming distances $\Hs{P}{T}{i}$ provided that $P$ and $T$ share a common $d$-period $p$. 
These values can be generated in $\Oh(\sqrt{(d+p)\log (d+p)})$ amortised time using $\Oh(d+p)$ words of space,
with $\Theta(d+p)$ consecutive Hamming distances actually computed in every iteration.
In \cref{lem:simpler}, we adapt this approach to the streaming setting, where $\Hs{P}{T}{i}$ needs to be known before $T[i+1]$ is revealed.
To deal with this, we use a two-part partitioning known as the \emph{tail trick}.
Similar ideas were already used to deamortise streaming pattern matching algorithms; see~\cite{CFPSS:2016,CFPSS:2015,CS:2010}. 

\begin{restatable}{lemma}{lemsimpler}\label{lem:simpler}
Let $P$ be a pattern with a $d$-period $p$.
Suppose that $p$ is also an $\Oh(d+p)$-period of the text $T$.
There exists a deterministic streaming algorithm which processes $T$ using $\Oh(d+p)$ words of space
and $\Oh(\sqrt{(d+p)\log (d+p)})$ time per symbol, and reports $\HS_{T,P}[i]$ for each position $i\ge |P|-1$. 
\end{restatable}

Our final goal in this section is to waive the assumption that $p$ is an approximate period of the text. 
Nevertheless, we observe that $p$ must still be a $(d+k)$-period of any fragment matching $P$ with $k$ mismatches.
Thus, our strategy is to identify approximately periodic fragments of $T$ which are guaranteed to
contain all $k$-mismatch occurrences of $P$; \cref{lem:simpler} is then called for each such fragment.
To make the result applicable in \cref{sec:main-algorithm},
we augment it with two extra features: 
First, we allow for delaying the output by $\Oh(|P|)$ positions, which is possible because
the approximately periodic fragments of $T$ can be stored in small space.
Secondly, we support reporting $\sk_k(T[0,\ldots,i-1])$ for any $k$-mismatch occurrence $T[i,\ldots,i+|P|-1]$.
This comes at the cost of $\Oh(\log^2 n)$ extra time per symbol due to the use of \cref{cor:streamsketch}.

\begin{restatable}{theorem}{thmsmallperiod}\label{thm:small-period}
Suppose that we are given an integer $k$ and the periodic representation of a pattern $P$
with respect to a $d$-period $p$ such that $d=\Oh(k)$ and $p=\Oh(k)$.
There exists a deterministic streaming algorithm, which uses $\Oh(k)$ words of space and $\Oh(\sqrt{k\log k} + \log^2{n})$ time per symbol
to report the $k$-mismatch occurrences of $P$ in the streamed text $T$.

For each reported occurrence, the mismatch information 
and sketch of the prefix of $T$ up to the reported occurrence can be computed on demand in $\Oh(k)$ and $\Oh(k\log^2 n)$ time,
respectively.
The algorithm may also be configured to report the output with any prescribed delay $\Delta=\Oh(|P|)$.
\end{restatable}

A combination of \cref{lem:perpref} and \cref{thm:small-period} now enables us to give a deterministic streaming $k$-mismatch algorithm when the pattern is guaranteed to have a $d$-period $p$ with $d=\Oh(k)$ and $p=\Oh(k)$.  
The procedure of \cref{lem:perpref} can be called to find an $\Oh(k)$-period $p'=\Oh(k)$ of $P$ (along with the periodic representation of $P$), and then $T$ can be processed using \cref{thm:small-period}.  
This concludes the description of our streaming $k$-mismatch algorithm in the  case where the pattern has at least one small approximate period.

\section{Efficiently Encoding Nearby $k$-Mismatch Occurrences}\label{sec:nearby}
In this section, we present the main technical contribution of our work:
we show how all $k$-mismatch occurrences of a pattern $P$ in a text $T$ of length $|T|=\Oh(|P|)$
can be stored in $\Oh(k(\log\frac{n}{k}+\log |\Sigma|))$ bits along with the underlying mismatch information.

First, we apply this tool to develop an optimal deterministic one-way communication protocol for Problem~\ref{prob:one},
where Bob must output all the alignments of the pattern and the text with at most $k$ mismatches, as well as the applicable mismatch information,
with no access to the text $T$ or the pattern $P$.
Next, we increase the space complexity to $\Oh(k\log n)$ bits, which lets us encode and retrieve
the $k$-mismatch occurrences using time- and space-efficient procedures.
The latter is a key building block of our $k$-mismatch streaming algorithm.

Our approach relies on the crucial observation that
overlapping $k$-mismatch occurrences of a pattern $P$
induce a $2k$-period of $P$.

\begin{fact}\label{fct:per}
	If $P$ has $k$-mismatch occurrences at positions $\ell,\ell$ of $T$ satisfying $\ell<\ell'<\ell+|P|$,
	then $\ell'-\ell \in \Per(P, 2k)$. 
\end{fact}

Consequently, we shall build a data structure 
that for a given string $X$ of length $n$
and a collection $\PP\sub\Per(X,k)$ of its $k$-periods lets us efficiently
encode all the underlying mismatch information $\MM(X[0,\ldots,n-p-1],X[p,\ldots,n-1])$ for $p\in \PP$.
It turns out that $\Oh(k(\log \frac{n}{k}+\log |\Sigma|))$ bits 
are sufficient provided that $\PP\sub \Per_{\le n/4}(X,k)$,
i.e., if each $k$-period $p\in \PP$ satisfies $p\le \frac14 n$.
Remarkably, the asymptotic size of the data structure matches the size of mismatch information of a single $k$-period despite encoding considerably more information.

Our data structure must in particular be able to retrieve $X[i]$
unless $X[i]=X[i+p]$ and $X[i]=X[i-p]$ hold for each $p\in \PP$ (such that $i+p < n$ and $i-p\ge 0$, respectively).
This idea can be conveniently formalised using the following concept of classes modulo $d$ in a string $X$.
\begin{definition}
	Let $X$ be a fixed string.
	For integers $i$ and $d$ with $d\ge 0$, the \emph{$i$-th class modulo $p$} (in $X$) is defined as
	a multiset:
	\[\C_{d}(X,i)=\{X[i'] : 1\le i \le |X| \text{  and  } i'\equiv i\!\!\! \pmod{p}\}.\]
	For $d=0$ we assume that $i' \equiv i \pmod{0}$ if and only if $i=i'$.
\end{definition}

Now, it is easy to see that we do not need to store $X[i]$ if the class
$\C_{d}(X,i)$ modulo $d=\gcd(\PP)$ is \emph{uniform},
i.e., if it contains just one element (with positive multiplicity).
Indeed, a mismatch $X[i]\ne X[i+p]$ or $X[i]\ne X[i-p]$ may only occur in a non-uniform class.
This is the motivation behind the following component developed in \cref{sec:key_proof}.

\begin{restatable}{lemma}{keylem}\label{lem:key}
	For a string $X$ and an integer $k$, let $\PP \sub \Per_{\le n/4}(X,k)$ and $d=\gcd(\PP)$.
	There is a data structure of size $\Oh(k(\log \frac{n}{k}+ \log |\Sigma|))$ bits which given an index $i$ 
	retrieves $X[i]$ in $\Oh(\log n)$ time if $\C_{d}(X,i)$ is non-uniform, and returns a sentinel symbol $\#$ if the class
	is uniform.
	The data structure can be initialised in $\Oh(1)$ time with $\PP=\emptyset$,
	and updated in $\Oh(k\log n)$ time subject to adding a $k$-period $p\in \Per_{\le n/4}(X,k)$ to $\PP$ given $\MM(X[0,\ldots,n-p-1],X[p,\ldots,n-1])$.
\end{restatable}

To prove \cref{lem:key}, we build a sequence $0=d_0,\ldots,d_s=d$ of integers such that $d_\ell = \gcd(d_{\ell-1},p_\ell)$,
$p_{\ell}\in \PP$ for $1\le \ell \le s$, and $s = \Oh(\log n)$. Next, we observe that classes modulo $d_{\ell}$
for $\ell=0,\ldots,s$ form a sequence of partitions of $\{X[i]: 0 \le i < n\}$, with each partition coarser than the previous one.
We keep the \emph{majority} of a class modulo $d_{\ell}$ whenever it differs form the majority of the enclosing class modulo $d_{\ell+1}$.
The majority of $\C_{d_{\ell}}(i)$ is likely to match the majority of $\C_{d_\ell}(i+p_{\ell})$, which lets us store these characters using \emph{run-length encoding}. 
On the top level, we keep the majority of every non-uniform class modulo $d_s$;
the number of such classes turns out to be $\Oh(k)$.

\subsection{An Optimal Deterministic Protocol for \cref{prob:one}}\label{sec:det}
Using \cref{lem:key}, we can now develop an efficient protocol for \cref{prob:one} and thus prove \cref{thm:prob-one}.   
	We shall assume that the text $T$ is of length at most $\frac54n$.
	If the actual text is longer, the full protocol splits it into substrings of length $\frac54n$
	with overlaps of length $n$,  thus enabling us to find all $k$-mismatch occurrences by repeating the protocol a constant number of times.
	
	If $P$ does not occur in $T$, Alice may send an empty message to Bob.
	Otherwise, her message consists of the following data:
	\begin{itemize}[itemsep=0pt,topsep=5pt,parsep=2pt]
		\item the locations $\ell$ and $\ell'$ of the leftmost and the rightmost $k$-mismatch occurrence of $P$ in $T$, along with the underlying mismatch information;
		\item the value $d=\gcd(\PP)$ and the data structure of \cref{lem:key}
		for $\PP$ consisting of distances between locations of $k$-mismatch occurrence of $P$ in $T$; $\PP\sub \Per_{\le n/4}(P, 2k)$ due to \cref{fct:per}.
	\end{itemize}
	By \cref{lem:key} and the tight bound on the size of the mismatch information, the message takes $\Oh(k(\log \frac{n}{k}+ \log |\Sigma|))$ bits.
	Now, it suffices to describe how Bob can retrieve all the $k$-mismatch occurrences of $P$ in $T$, as well as the corresponding mismatch information.
	
	For this, we show that Bob can construct a \emph{proxy pattern} $P_\#$ as well as a \emph{proxy text} $T'_\#$  
	and look for $k$-mismatch occurrences of $P_\#$ in $T'_\#$ instead of the $k$-mismatch occurrences of $P$ in $T$.
	More precisely, we first obtain $T'=T[\ell,\ldots,\ell'+|P|-1]$ from $T$
	by trimming the prefix and the suffix of $T$ disjoint with the $k$-mismatch occurrences of $P$.
	Next, we construct $P_\#$ from $P$
	by replacing $P[i]$ with a sentinel $\#_{i\bmod d}$ (distinct between classes modulo $d$) whenever $\C_d(P,i)\cup \C_{d}(T',i)$ is uniform.
	Similarly, we transform $T'$ to obtain $T'_\#$.
	
	The main property of these strings is that in any occurrence of $P$ in $T'$, we have not altered the symbols involved in a mismatch,
	while matching symbols could only be replaced by sentinels in a consistent way.
	\begin{restatable}{fact}{fctequiv}\label{fct:equiv}
		The pattern $P$ has a $k$-mismatch occurrence at position $j$ of $T$
		if and only if $P_\#$ has a $k$-mismatch occurrence at position $j-\ell$ of $T'_\#$.
		Moreover, in that case we have
		\[\MM(P,T[j,\ldots,j+n-1])=\MM(P_\#,T'_\#[j-\ell,\ldots,j-\ell+n-1]).\]
	\end{restatable}
	
	Furthermore, we deduce that Bob can retrieve any symbol of $P_\#$	
	based on \cref{lem:key} (if $\C_d(P,i)$ is non-uniform) or the mismatch information for the $k$-mismatch occurrences of $P$
	as a prefix and a suffix of $T$ (otherwise). Similarly, Bob can retrieve $T'_\#$ because $T'$ is covered
	by the two $k$-mismatch occurrences of $P$.
	\begin{restatable}{fact}{fctretr}\label{fct:retr}
		Bob can retrieve $P_\#$ and $T'_\#$ from Alice's message.
	\end{restatable}
	
	Consequently, Bob's strategy is to compute the mismatch information for all the alignments of $P_\#$ in $T'_\#$
	and output it  (with the starting position shifted by $\ell$) whenever there are at most $k$ mismatches.
	By \cref{fct:equiv}, this coincides with the desired output of $k$-mismatch occurrences of $P$ in $T$.
	This concludes the proof of \cref{thm:prob-one}, whose statement is repeated below.

\thmone*

\subsection{Algorithmic Consequences}\label{sec:rand}
In this section, we apply the ideas behind \cref{thm:prob-one} to develop time- and space-efficient compression scheme for the following representation of the $k$-mismatch occurrences of $P$ in $T$.
\begin{definition}
Consider a pattern $P$ and a text $T$.
We define the \emph{stream of $k$-mismatch information} of $P$ in $T$ as a sorted sequence
of starting positions $i$ of $k$-mismatch occurrences of $P$ in $T$,
each associated with the mismatch information and the sketch $\sk_k(T[0,\ldots,i-1])$.
\end{definition}
 Note that the inclusion of sketches increases the space consumption
of a single entry to $\Theta(k\log n)$ bits compared to $\Theta(k(\log \frac{n}{k}+\log |\Sigma|))$ bits required
for the mismatch information.

Let us see how to implement fast procedures for both parties of \cref{prob:one}.
As far as Alice is concerned, we observe that \cref{fct:per} can be trivially made constructive.
In other words, it is easy to transform the mismatch information of two overlapping $k$-mismatch occurrences of $P$
to the mismatch information of the induced $2k$-period of $P$.
As stated in \cref{lem:key}, the underlying component can be efficiently constructed based on this data.

To implement Bob's procedure, we observe that the proof of \cref{fct:equiv}
actually allows for read-only random access to $P_\#$ and $T'_\#$, with $\Oh(\log n)$ time required to retrieve any symbol.
Thus, we develop an efficient $k$-mismatch algorithm for that model, and we use it to locate $k$-mismatch occurrences of $P_\#$ in $T'_\#$.
This procedure is based on the results of \cref{sec:periodic} applied for the longest prefix $Q$ of $P$
with a $(2k+1)$-period $p\le k$, with the sketches of \cref{sec:fingerprint} employed to check
which $k$-mismatch occurrences of $Q$ extend to $k$-mismatch occurrences of $P$.

\begin{restatable}{theorem}{thmram}\label{thm:ram}
In the read-only random-access model, the streaming $k$-mismatch problem can be solved
on-line with a Monte-Carlo algorithm 
using $\Oh(k\log n)$ bits of working space and $\Oh(\sqrt{k\log k}+\log^3 n)$ time
per symbol, including $\Oh(1)$ symbol reads.
For any reported $k$-mismatch occurrence, the mismatch information can be retrieved
on demand in $\Oh(k)$ time.
\end{restatable}

Compared to the output of \cref{prob:one}, the stream of $k$-mismatch information also includes the sketches
$\sk_k(T[0,\ldots,i-1])$ for each $k$-mismatch occurrences $T[i,\ldots,i+|P|-1]$.
We use \cref{cor:streamsketch,lem:manip} to retrieve them efficiently after some non-trivial preprocessing.

Our final observation is that several instances of the resulting data structure
can be combined to form a \emph{buffer} allowing us to \emph{delay} the stream of $k$-mismatch information
by a prescribed value $\Delta$.
This component is defined in a synchronous setting, based on an external clock which measures the progress of processing the text $T$,
with discrete \emph{ticks} corresponding to scanning subsequent symbols of $T$.
If $P$ has a $k$-mismatch occurrence at position $i$, then the buffer is fed with the entry of the $k$-mismatch stream at tick $i$,
and it should report it back at tick $i+\Delta$. 

\begin{restatable}{theorem}{thmbuffer}\label{thm:buffer}
The stream of $k$-mismatch information of $P$ in $T$ can delayed by any $\Delta=\Theta(|P|)$ using a buffer of $\Oh(k \log n)$ bits,
which takes $\Oh(\sqrt{k\log k}+\log^3 n)$ time per tick, with $\Oh(k\log^2 n)$ extra time for ticks when 
it is fed with a $k$-mismatch occurrence of $P$ or it reports one.
The initialisation, given $\sk_k(P)$ and $\Delta$, takes $\Oh(k)$ time.
\end{restatable}

\section{The Streaming  $k$-mismatch Algorithm}\label{sec:main-algorithm}
We can now present our main result, an \ourcomplexity-time and $\Oh(k\log{n}\log\frac{n}{k})$-bit streaming algorithm for the $k$-mismatch problem. 
Without loss of generality, we assume that $|P|\ge k$ so that no output is required while we read the first $k$ symbols of the text.
As a result, we can start processing the text with a delay of $k/2$ symbols
and catch up while processing the subsequent $k/2$ symbols of the text.
This allows for a small amount of post-processing time once we finish reading the pattern.

\subsection{Processing the Pattern}
Let us first describe the information about the pattern that we gather.
We distinguish between two cases. 
If we discover that the pattern $P$ has an $\Oh(k)$-period $p\le k$,
then we can store $P$ using the periodic representation with respect to $p$.

Otherwise, we proceed in a similar fashion to the work on streaming exact matching~\cite{Porat:09,BG:2014},
and we introduce a family of $\Oh(\log n)$ prefixes $P_0,\ldots,P_L$ of $P$. We choose it so that:
\begin{itemize}[itemsep=0pt,topsep=5pt,parsep=2pt]
  \item $P_0$ is the longest prefix of $P$ with a $(2k+1)$-period $p\le k$,
  \item $|P_{L}|=n=|P_{L-1}|+2k$, and
  \item $|P_{1}|,\ldots,|P_{L-2}|$ are the subsequent powers of two between $|P_0|$ and $\frac12 |P_{L-1}|$ (exclusive).
\end{itemize}
We store the periodic representation of $P_0$,
the sketches $\sk_k(P_\ell)$ for $1 \le \ell < L$,
and the trailing $2k$ symbols of $P$.

\begin{lemma}
The pattern $P$ can be processed by a deterministic streaming algorithm
which uses $\Oh(k\log n \log\frac{n}{k})$ bits of space, takes $\Oh(\sqrt{k\log k}+\log^2 n)$ time per symbol,
and $\Oh(k\sqrt{k\log k} + k\log^2 n)$ post-processing time.
\end{lemma}
\begin{proof}
While scanning the pattern, we store a buffer of $2k$ trailing symbols and we 
run the algorithm of \cref{lem:perpref} to compute the longest prefix $P_0$ of $P$ with a $(2k+1)$-period $p\le k$.
If $|P_0|>n-2k$, we extend the periodic representation of $P_0$ to the periodic representation of $P$,
for which $p$ must be a $4k$-period.

Otherwise, we set $P_1,\ldots,P_L$ as specified above. 
To have $\sk_k(P_\ell)$ available,
we use \cref{cor:streamsketch} while scanning $P$ and extract the sketches of every prefix of $P$
whose length is a power of 2 larger than $3k$ (note that $|P_0|\ge 3k$). We run this subroutine
with a delay of $2k$ symbols so that $\sk_k(P_{L-1})$ can be retrieved as soon as the whole pattern $P$ is read.
\end{proof}

\subsection{Processing the Text}
In the small approximate period case, we simply use \cref{thm:small-period} (with no delay).

Otherwise, our algorithm is organised into several \emph{levels} $\ell=0,\ldots,L$.
The task of level $\ell$ is to output the stream of $k$-mismatch information of $P_\ell$ in $T$. 
This stream is forwarded to the next level $\ell+1$ or to the output (for $\ell=L$).
For $\ell < L-1$, the stream shall be generated with a delay of $|P_{\ell+1}|-|P_{\ell}|$,
i.e., the occurrence at position $i$ shall be reported just before the algorithm reads $T[i+|P_{\ell+1}|]$.
For $\ell=L-1$, the delay is specified as $k$, while the topmost level $\ell=L$ must report the $k$-mismatch occurrences of $P$ with no delay.

Level $0$ is implemented using \cref{thm:small-period} with a delay of $|P_1|-|P_0| = \Oh(|P_0|)$ symbols.
The running time is $\Oh(\sqrt{k\log k}+\log^2 n)$ per symbol.

The implementation of level $\ell$, $0 < \ell < L$ is based on \cref{cor:streamsketch,thm:buffer}.
When a $k$-mismatch occurrence of $P_{\ell-1}$ is reported
at position $j$, we use \cref{cor:streamsketch} to obtain $\sk_k(T[0,\ldots,j+|P_{\ell}|-1])$.
We also retrieve $\sk_k(T[0,\ldots,j-1])$ from the stream of $k$-mismatch information of $P_{\ell-1}$, which lets us derive $\sk_k(T[j,\ldots,j+|P_{\ell}|-1])$
using \cref{lem:manip}. Then, we use \cref{cor:sketches} to test if $T[j,\ldots,j+|P_{\ell}|-1]$ is a $k$-mismatch occurrence of $P_{\ell}$
and to retrieve the mismatch information if the answer is positive.
These computations take $\Oh(k\log^3 n)$ and we perform them while processing the subsequent $k$ symbols of $T$.
By \cref{fct:per}, we are guaranteed that just one $k$-mismatch occurrence of $P_{\ell-1}$ is processed at any given time.
Thus, the stream of $k$-mismatch occurrences for the pattern $P_{\ell}$ in the text $T$ can be generated with delay $k$.
For $\ell < L-1$, we apply a buffer of \cref{thm:buffer} to delay it further by $|P_{\ell+1}|-|P_{\ell}|-k$.
This comes at the extra cost of $\Oh(\sqrt{k\log k}+\log^3 n)$ time per symbol due to $|P_{\ell+1}|-|P_{\ell}|-k=\Theta(|P_{\ell}|)$:
\[4|P_{\ell}| \ge |P_{\ell+1}| \ge |P_{\ell+1}|-|P_{\ell}|-k \ge |P_{\ell}|-k \ge  |P_{\ell}|-\tfrac13|P_0| \ge \tfrac23 |P_{\ell}|.\]

Finally, the topmost level $L$ 
simply maintains the trailing $2k$ symbols of the text.
Whenever a $k$-mismatch occurrence of $P_{L-1}$ at position $i$ is reported (along with the mismatch information),
we must check if it extends to a $k$-mismatch occurrence of $P=P_L$.
It arrives with delay $k$, so we can naively compare $P[n-2k,\ldots,n-1]$
with $T[i+n-2k,\ldots,i+n-1]$ while the algorithm processes $T[i-n-k,\ldots,i-1]$,
extending the mismatch information accordingly.
Space complexity of level $L$ is $\Oh(k\log n)$ and the per-symbol running time is constant.

Aggregating the resources required for each of the $\Oh(\log \frac{n}{k})$ levels,
we obtain our main result.

\thmstreaming*

\section{A Conjectured Space Lower Bound}

We conclude this first part of the paper with a discussion of space lower bounds for the streaming $k$-mismatch problem.  The space upper bound we give for our streaming algorithm, although close to being optimal,  is still an $\Oh(\log n)$-factor away from the known lower bound.  As a final contribution, we give a higher conjectured space lower bound for Problem~\ref{prob:streaming}, which partially closes this gap.  We do this by observing that a particularly natural way to tackle the streaming problem is first to encode the $k$-mismatch alignments and mismatch information for all prefixes of the pattern against all suffixes of a substring of the text of the same length and then use this information to start processing new symbols as they arrive.   Our streaming algorithm, for example, effectively does exactly this, as do the earlier streaming $k$-mismatch algorithms of~\cite{Porat:09,CFPSS:2016} and the exact matching streaming algorithms of~\cite{Porat:09,BG:2014}. We show a space lower bound for Problem~\ref{prob:streaming}  for any streaming algorithm that takes this approach.  We further conjecture that this  lower bound is, in fact, tight in general.

\begin{conjecture}\label{conj:prob-periods}
	Any solution for Problem~\ref{prob:streaming}  must use at least $\Omega( k  \log\frac{n}{k}(\log \frac{n}{k} + \log{|\Sigma|}))$  bits  of space.
\end{conjecture}

%

We argue in favour of this conjecture by giving an explicit set of patterns for which encoding the mismatch information for all alignments of the pattern against itself will require the stated number of bits.   If the text includes a copy of the pattern  as a substring, then the result follows.     It is interesting to note that a similar conjecture can be made for the space complexity of exact  pattern matching in a stream. In this case, encoding the alignments of all exact matches between the prefixes and suffixes of a pattern is known to require $\Omega(\log^2{n})$ bits, matching the best known space upper bounds of~\cite{Porat:09,BG:2014}.


\begin{lemma}
	Any solution for Problem~\ref{prob:streaming} that computes all alignments with Hamming distance at most $k$ between prefixes of the pattern and equal-length suffixes of the text, along with the associated mismatch information,  must use at least $\Omega( k  \log\frac{n}{k}(\log \frac{n}{k} + \log{|\Sigma|}))$  bits  of space.
\end{lemma}
\begin{proof}
	We prove the space lower bound by showing  an explicit set of patterns for which any encoding of all the $k$-mismatch alignments between prefixes of the pattern $P$ and suffixes of the text $T=P$, along with the mismatch information, must use $\Omega( k  \log\frac{n}{k}(\log\frac{n}{k} + \log{|\Sigma|}))$ bits.  We define our string recursively.  Consider a base string $S_0 = 0^k$.  To create $S_{i+1}$, we make three copies of $S_i$ and concatenate them to each other to make $S_iS_iS_i$.  We then choose $\lfloor \frac12 k \rfloor$ indices  at random from the middle copy of $S_i$ and randomly change the symbols at those indices. Let $S'_i$ be this modified middle third so that $S_{i+1} = S_iS'_iS_i$.    There are $\Theta(\log\frac{n}{k})$ levels to the recursion, so in the final string we can identify $\Omega(\log\frac{n}{k})$ alignments at which the Hamming distance is at most $k$:
the prefix $S_i S'_i$ and the suffix $S'_i S_i$ are at Hamming distance $2\lfloor\frac12 k \rfloor$.  Mismatch information for each of these alignments contains  $\lfloor \frac12 k \rfloor$ randomly chosen indices and  $\lfloor \frac12 k \rfloor$ random symbols which need to be reported.  This gives a lower bound of  $\Omega( k  \log\frac{n}{k}(\log\frac{n}{k} + \log{|\Sigma|}))$  bits in total needed for any encoding.
\end{proof}

\section{Algebraic Algorithms on $\Fp$}

In this section, we recall several classic problems in computer algebra involving a prime field $\Fp$, which arise in \cref{sec:fingerprint}.
The time and space complexities of their solutions depend on the relation between
the field size $p$ and the input size, as well as on the model of computation.
Below, we state these complexities for the setting used throughout the paper, which is as follows:
We assume the word RAM model with word size $w$,
which supports constant-time arithmetic and bitwise operations on $w$-bit integers.

\begin{lemma}[Integer multiplication; Sch{\"{o}}nhage and Strassen~\cite{SS:1971,Knuth:1997}]
The product of two $n$-bit integers can be computed in $\Oh(n)$ time using $\Oh(n)$ bits of space
provided that $w=\Omega(\log n)$.
\end{lemma}

Next, we consider operations in the field $\Fp$ where $p$ is a prime number with $\log p = \Theta(w)$,
and the corresponding ring of polynomials $\Fp[X]$. We always assume that the degrees of input polynomials are bounded by $n = 2^{\Oh(w)}$.
Polynomial multiplication is a basic building block of almost all efficient algebraic algorithms on $\Fp$.
Our model of computation allows for the following efficient solution:
\begin{cor}[Polynomial multiplication]\label{cor:mult}
Given two polynomials $A,B \in \Fp[X]$ of degree at most $n$, the product $A\cdot B$ can be computed
in $\Oh(n\log p)$ time using $\Oh(n\log p)$ bits of space.
\end{cor}
\begin{proof}
Polynomial multiplication in $\mathbb{Z}[X]$ can be reduced to integer multiplication via the Kronecker substitution~\cite{Kronecker1882}.
More precisely, a polynomial $P(X)=\sum_{i=0}^n p_i X^i$ of degree $n$ with $0\le p_i < N$
 is represented as an integer with $n+1$ blocks of $1+2\log N +\log n$ bits each
so the binary encoding of $p_i$ is stored in the $i$-th least significant block.
To multiply polynomials in $\Fp[X]$, we can compute the product in $\mathbb{Z}[X]$ and then replace each coefficient by its remainder modulo $p$.
\end{proof}

In \cref{sec:fingerprint}, we use efficient solutions to three classic problems listed below.
The original papers refer provide the space complexity in terms of the time $M(n)$ of polynomial multiplication in $\Fp$.
The space complexity is not specified explicitly; one can retrieve it by analysing the structure of the original algorithms
and their subroutines, such as multi-point evaluation, polynomial division, and $\gcd$ computation of polynomials.
We refer to a textbook~\cite{DBLP:books/daglib/0031325} for detailed descriptions of these auxiliary procedures as well as of the Cantor--Zassenhaus algorithm.

\begin{lemma}[Polynomial factorization; Cantor--Zassenhaus~\cite{CZ:1981}]\label{lem:cz}
Given a polynomial $A\in \Fp[X]$ of degree $n$ with $n$ distinct roots,
all the roots of $A$ can be identified in $\Oh(n \log^3 p)$ time using $\Oh(n\log p)$ bits of space.
The algorithm may fail (report an error) with probability inverse polynomial in $p$.
\end{lemma}
\begin{proof}
The Cantor--Zassenhaus algorithm proceeds in several iterations; see \cite{DBLP:books/daglib/0031325}.
Each iteration involves $\log p$ multiplications and divisions of degree-$\Oh(n)$ polynomials,
as well as several $\gcd$ computations involving polynomials of total degree $\Oh(n)$.
The overall time of these operations is $\Oh(n \log^2 p + n\log n \log p) = \Oh(n\log^2 p)$.
Each step can be interpreted as a random partition of the set of roots of $A$ into two subsets. The computation terminates
when every two roots are separated by at least one partition. After $\Omega(\log p)$ phases
this condition is not satisfied with probability inverse polynomial in $p$.
\end{proof}

\begin{lemma}[BCH Decoding; Pan~\cite{Pan:1997}]\label{lem:bch}
Consider a sequence $s_j = \sum_{i=0}^{n-1} \alpha_i \cdot \beta_i^j \in \Fp$ with distinct values $\beta_i$ and coefficients $\alpha_i\ne 0$ for $n < p$. 
Given the values $s_0, s_1,\ldots,s_{2n}$, the polynomial $\prod_{i=1}^{n}(1-X\beta_i)$ can be computed in $\Oh(n \log p)$ time using $\Oh(n \log p)$ bits of space.
\end{lemma}

\begin{lemma}[Transposed Vandermonde matrix-vector multiplication; Canny--Kaltofen--Lakshman~\cite{DBLP:conf/issac/CannyKY89}]\label{lem:vand_eval}
Consider a sequence $s_j = \sum_{i=0}^{n-1} \alpha_i \cdot \beta_i^j \in \Fp$ with distinct values $\beta_i\in \Fp$ and arbitrary coefficients $\alpha_i\in \Fp$.
Given the values $\alpha_0,\ldots, \alpha_{n-1}$ and $\beta_0,\ldots,\beta_{n-1}$, the coefficients $s_0,\ldots,s_{n-1}$
can be computed in $\Oh(n \log n \log p)$ time using $\Oh(n\log p)$ bits of space.
\end{lemma}

\begin{lemma}[Solving transposed Vandermonde systems; Kaltofen--Lakshman~\cite{DBLP:conf/issac/KaltofenY88}]\label{lem:vand}
Consider a sequence $s_j = \sum_{i=0}^{n-1} \alpha_i \cdot \beta_i^j \in \Fp$ with distinct values $\beta_i\in \Fp$ and arbitrary coefficients $\alpha_i\in \Fp$.
Given the values $s_0,\ldots,s_{n-1}$ and $\beta_0,\ldots,\beta_{n-1}$, the coefficients $\alpha_0,\ldots,\alpha_{n-1}$ can be retrieved in $\Oh(n \log n \log p)$ time
using $\Oh(n\log p)$ bits of space.
\end{lemma}

\section{Omitted Proofs from \cref{sec:fingerprint}}\label{sec:fingerprint2}

\corsketches*
\begin{proof}
  First, suppose that $\HD(S,T)=k' < k$.
	 Let $x_1,\ldots, x_{k'}$ be the mismatch positions of $S$ and $T$,
	and let $r_i = S[x_i] - T[x_i]$ be the corresponding numerical differences.
	We have:
	\[
	\begin{array}{ccccccccc}
	r_1 & + & r_2 & + & \ldots & + & r_{k'} & = & \phi_0(S)-\phi_0(T) \\
	r_1 x_1 & + & r_2 x_2 &  + & \ldots & + & r_{k'} x_{k'}  & = & \phi_{1}(S)-\phi_1(T)     \\
	r_1 x_1^2 & + & r_2 x_2^2 &  + & \ldots & + & r_{k'} x_{k'}^2 & =  &  \phi_{2}(S)-\phi_2(T)     \\
	&   &       & \vdots & & & & & \vdots \\
	r_1 x_1^{2k} & + & r_2 x_2^{2k} & + & \ldots & + & r_{k'} x_{k'}^{2k} & =  &  \phi_{2k}(S) -\phi_{2k}(T)
	\end{array}
	\]
	
	This set of equations is similar to those appearing in~\cite{CEPR:2009} and in the decoding procedures for Reed--Solomon codes.
	We use the standard Peterson--Gorenstein--Zierler procedure~\cite{DBLP:journals/tit/Peterson60,doi:10.1137/0109020},
	with subsequent efficiency improvements.
	This method consists of the following main steps:
	\begin{enumerate}[itemsep=0pt,topsep=5pt,parsep=2pt]
		\item\label{step:coeffs} Compute the \emph{error locator polynomial} $P(z) = \prod_{i=1}^{k'} (1-x_i z)$ from the $2k+1$ \emph{syndromes} ${\phi_{j}(S)-\phi_j(T)}$
		with $0\le j \le 2k$.
		\item\label{step:rootfinding} Find the \emph{error locations} $x_i$ by factoring the polynomial $P$.
		\item\label{step:errorrecovery} Retrieve the \emph{error values} $r_i$.
	\end{enumerate}
	
	We implement the first step in $\Oh(k\log n)$ time using the  efficient \emph{key equation} solver by Pan~\cite{Pan:1997}; see \cref{lem:bch}.
	The next challenge is to factorise $P$, taking advantage of the fact that it is a product of linear factors.
	As we are working over a field with large characteristic, there is no sufficiently fast deterministic algorithm for this task. 
	Instead we use the randomised Cantor--Zassenhaus algorithm~\cite{CZ:1981} (see~\cref{lem:cz}), which takes $\Oh(k \log^3 n)$ time with high probability. 
	If the algorithm takes longer than this time, then we stop the procedure and report a failure.
	Finally, we observe that the error values $r_i$ can be retrieved by solving a transposed Vandermonde linear system
	of $k'$ equations using Kaltofen--Lakshman algorithm~\cite{DBLP:conf/issac/KaltofenY88} (see \cref{lem:vand}) in $\Oh(k\log k \log n)$ time.
	Each of these subroutines uses $\Oh(k \log n)$ bits of working space.
	
	Using the fact that we now have full knowledge of the mismatch indices $x_i$, a similar linear system lets us retrieve the values $r'_i = S^2[x_i] - T^2[x_i]$:
	\[\begin{array}{ccccccccc}
	r'_1 & + & r'_2 & + & \ldots & + & r'_{k'} & = & \phi'_0(S)-\phi'_0(T) \\
	r'_1 x_1 & + & r'_2 x_2 &  + & \ldots & + & r'_{k'} x_{k'}  & = & \phi'_{1}(S)-\phi'_1(T)     \\
	r'_1 x_1^2 & + & r'_2 x_2^2 &  + & \ldots & + & r'_{k'} x_{k'}^2 & =  &  \phi'_{2}(S)-\phi'_2(T)     \\
	&   &       & \vdots & & & & & \vdots \\
	r'_1 x_1^{k'} & + & r'_2 x_2^{k'} & + & \ldots & + & r'_{k'} x_{k'}^{k'} & =  &  \phi'_{k'}(S) -\phi'_{k'}(T)
	\end{array}
	\]
	Now, we are able to compute $S[x_i] = \frac{r'_i+r_i^2}{2r_i}$ and $T[x_i] = \frac{r'_i-r_i^2}{2r_i}$.

If $\HD(S,T)>k$, then we may still run the procedure above, but its behaviour is undefined.
This issue is resolved by using the Karp--Rabin fingerprints to help us check if we have found all the mismatches or not.

 If the algorithm fails, we may assume that $\HD(S,T) > k$;
	otherwise, the failure probability is inverse polynomial in $n$.

	A successful execution results in the mismatch information $\{(x_i, s_i, t_i) : 1\le i \le k'\}$.
	Observe that $\HD(S,T)\le k$ if and only if $S[x_i]-T[x_i]=s_i-t_i$ and $S[j]-T[j]=0$ at the remaining positions.
	In order to verify this condition, we compare the Karp--Rabin fingerprints, i.e., test whether
	\[\psi_r(S)-\psi_r(T) = \sum_{i=1}^{k'} (s_i-t_i) r^{x_i}.\]
	This verification takes $\Oh(k' \log n)$ time and its error probability is at most $\frac{\ell}{p}$.
\end{proof}

\lemconstr*
\begin{proof}
Let $\MM(S,T)=\{(x_i, s_i,t_i) : 0\le i < d\}$.
First, observe that the Karp--Rabin fingerprint can be updated in $\Oh(d\log n)$ time.
Indeed, we have $\psi_r(T)-\psi_r(S)=\sum_{i=0}^{d-1} (t_i-s_i)r^{x_i}$, and each power $r^{x_i}$ can be computed in $\Oh(\log n)$ time.
Next, we shall compute $\phi_j(T)-\phi_j(S)=\sum_{i=0}^{d-1} (t_i-s_i)x_i^j$
for $j\le 2k$. This problem is an instance of transposed Vandermonde evaluation; see \cref{lem:vand_eval}.
Hence, this task can be accomplished in $\Oh((d+k)\log^2 n)$ time using $\Oh((d+k)\log n)$ bits of space
using the Canny--Kaltofen--Lakshman algorithm.
The values $\phi'_j(T)-\phi'_j(S)$ are computed analogously.
\end{proof}

\corstreamsketch*
\begin{proof}
First, observe that appending $0$ to $X$ does not change the sketch, so appending a symbol
$a$ is equivalent to substituting $0$ at position $|X|$ by $a$.
Thus, below we focus on substitutions.

We maintain the sketch $\sk_k(Y)$ of a previous version $Y$ of $X$ and a buffer of up to $k$
substitutions required to transform $Y$ into $X$, i.e., the mismatch list $\MM(X,Y)$.
If there is room in the buffer, we simply append a new entry to $\MM(X,Y)$   to handle a substitution.
Whenever the buffer becomes full, we apply \cref{lem:constr} to compute $\sk_k(X)$ based 
on $\sk_k(Y)$ and $\MM(X,Y)$. This computation takes $\Oh(k\log^2n)$ time, so we run in it parallel
to the subsequent $k$ updates (so that the results are ready before the buffer is full again).

To implement a query, we complete the ongoing computation of $\sk_k(Y)$ and use \cref{lem:constr}
again to determine $\sk_k(X)$ (and clear the buffer as a side effect). This takes $\Oh(k\log^2 n)$ time.
\end{proof}

\lemmanip*
\begin{proof}
\textbf{1.} 	First, observe that $\psi_r(UV)=\psi_r(U)+r^{|U|}\psi_r(V)$. This formula can be used to retrieve one of the Karp--Rabin fingerprints
	given the remaining two ones. The running time $\Oh(\log n)$ is dominated by computing $r^{|U|}$ (or $r^{-|U|}$).
	Next, we express $\phi_j(UV)-\phi_j(U)$ in terms of $\phi_{j'}(V)$ for $j'\le j$:
	\[
	\phi_j(UV)-\phi_j(U) = \sum_{i=0}^{|V|-1} V[i](|U|+i)^j = \sum_{i=0}^{|V|-1}\sum_{j'=0}^j V[i]\binom{j}{j'}i^{j'} |U|^{j-j'}=\sum_{j'=0}^j \binom{j}{j'}\phi_{j'}(V)|U|^{j-j'}.
	\]
	Let us introduce an exponential generating function $\Phi(S) = \sum_{j=0}^\infty \phi_j(S)\frac{X^j}{j!}$,
	and recall that the exponential generating function of the geometric progression with ratio $r$ is $e^{rX}=\sum_{j=0}^\infty r^j\frac{X^j}{j!}$.
	Now, the equality above can be succinctly written as $\Phi(UV)-\Phi(U)= \Phi(V)\cdot e^{|U|X}.$
	Consequently, given first $2k+1$ coefficients of $\Phi(U)$, $\Phi(V)$, or $\Phi(UV)$, can be computed from the first $2k+1$ terms of the other two generating functions
	in $\Oh(k\log n)$ time using efficient polynomial multiplication over $\Fp$~\cite{SS:1971}; see \cref{cor:mult}.
	The coefficients $\phi'_j(U)$, $\phi'_j(V)$, and $\phi'_j(UV)$, can be computed in the same way.

\noindent\textbf{2.} Observe that $\psi_r(U^m) = \sum_{i=0}^{m-1} r^{i|U|} \psi_r(U) = \frac{r^{m|U|}-1}{r^{|U|}-1} \psi_r(U)$.
Thus, $\psi_r(U)$ and $\psi_r(U^m)$ are easy to compute from each other in $\Oh(\log n)$ time.
Next, recall that the exponential generating function $\Phi(S) = \sum_{j=0}^\infty \phi_j(S)\frac{X^j}{j!}$
satisfies $\Phi(UW)=\Phi(U)+e^{|U|X}\Phi(W)$. Consequently, $\Phi(U^m)=\Phi(U)\cdot \sum_{i=0}^{m-1}e^{i|U|X}=\Phi(U^m)\frac{e^{m|U|X}-1}{e^{|U|X}-1}$.
Thus, $\Phi(U)$ and $\Phi(U^m)$ can be computed from each other in $\Oh(k\log n)$ time using efficient polynomial multiplication over $\Fp$; see~\cref{cor:mult}.
The first $\Oh(k)$ terms of the inverse of the power series $(e^{\ell X}-1)/X$ can also be computed in $\Oh(k\log n )$ time using polynomial multiplication
and Newton's method for polynomial division.
\end{proof}

\section{Omitted Proofs from \cref{sec:periodic}}

\lempp*

\begin{proof}
The constraints on the time and space complexity let us process $X$ in blocks of $p$ symbols (the final block might be shorter),
spending $\Oh(p\sqrt{p\log p})$ time on each block.

After reading each block, we compute the sought longest prefix $Y$ of $X$
with a $d$-period $p'\le p$,
the periodic representation of $Y$ with respect to $p'$,
and the values \[\HD(Y[0,\ldots,|Y|-p''-1], Y[p'',\ldots,|Y|-1])\] for $0<p''\le p$.
When at some iteration we discover that $Y$ is a proper prefix of $X$,
then $Y$ cannot change anymore, so we can ignore any forthcoming block.

Thus, below we assume that $X=Y$ before a  block $B$ is appended to $X$.
In this case, we need to check if $B$ can be appended to $Y$ as well, i.e., if $YB$ has any $d$-period $p''\le p$.
We rely on the formula 
 \begin{multline*}\HD((YB)[0,\ldots,|YB|-p''-1], (YB)[p'',\ldots,|YB|-1])=\\=\HD(Y[0,\ldots,|YB|-p''-1], Y[p'',\ldots,|Y|-1])+
 \HD((YB)[|Y|-p'',\ldots,|YB|-p''-1],B).\end{multline*}
 The left summand is already available, while to determine the right summand for each $p''$,
 we use Abrahamson's algorithm~\cite{DBLP:journals/siamcomp/Abrahamson87} to compute the Hamming distance
 of every alignment of $B$ within the the suffix of $YB$ of length $p+|B|$.
 This procedure takes $\Oh(p\sqrt{|B|\log |B|})$ time.
 
If we find out that $YB$ has a $d$-period $p''\le p$,
we compute the periodic representation of $Y$ with respect to $p''$.
For this, we observe that $Y[i]\ne Y[i-p'']$ may only hold if $i<p'+p''$,
$Y[i]\ne Y[i-p']$, $Y[i-p']\ne Y[i-p'-p'']$, or $Y[i-p'-p'']\ne Y[i-p'']$.
Thus, it takes $\Oh(d+p)=\Oh(p)$ time to transform the periodic representation of $Y$
with respect to $p'$ to the one with respect to $p''$.
Finally, we append $B$ to $Y$ and update its periodic representation.

Otherwise, we partition $B$ into two halves $B=B_LB_R$ and try appending the left half $B_L$ to $Y$
using the procedure above. We recurse on $B_R$ or $B_L$ depending on whether it succeeds.

The running time of the $i$th iteration is $\Oh\big(p+p\sqrt{\frac{p}{2^i} \log \frac{p}{2^i}}\big)$,
because we attempt appending a block of length at most $\lceil\frac{p}{2^i}\rceil$.
Consequently, the overall processing time is $\Oh(p\sqrt{p\log p})$.
\end{proof}

\fctotimesham*
\begin{proof}
	If $|P|-1 \le i  < |T|$, then:
	\begin{multline*}
	|P|-\HS_{P,T}[i] = |P|-\sum_{j=0}^{|P|-1}[T[i-j]\ne P[|P|-1-j]] = \sum_{j=0}^{|P|-1}[T[i-j]=P[|P|-1-j]]=\\
	=\sum_{a\in\Sigma} \sum_{j=0}^{|P|-1}T_{a}(i-j)P_a(|P|-1-j)=\sum_{a\in\Sigma} \sum_{j=0}^{|P|-1}T_{a}(i-j)P^R_a(j)=
	\sum_{a\in\Sigma}(T_a*P_a^R)(i)=(T\otimes P)(i).
	\end{multline*}
	The second claim follows from the fact that $X_a(i)=0$ if $i<0$ or $i\ge |X|$.
\end{proof}

\lemfft*
\begin{proof}
	Let us partition $\mathbb{Z}$ into blocks  $B_k = [\delta k, \delta k+\delta)$ for $k\in \mathbb{Z}$.
	Moreover, we define $B'_k=(i-(k+1)\delta, i-(k-1)\delta]$ and observe that $(f*g)(j) = \sum_{k} (f|_{B_k}* g|_{B'_k})(j)$
	for $i \le j < i + \delta$.
	
	We say that a block is \emph{heavy} if $|B_k \cap \mathrm{supp}(f)|\ge \sqrt{\delta \log \delta}$,
	i.e., if at most $\sqrt{\delta\log \delta}$ non-zero entries of $f$ belong to $B_k$.
	For each heavy block $B_k$, we compute the convolution of $f|_{B_k}*g|_{B'_k}$ 
	using the Fast Fourier Transform. This takes $\Oh(\delta\log \delta)$ time per heavy block and $\Oh(n\sqrt{\delta\log \delta})$ in total.
	
	The light blocks $B_k$ are processed naively: we iterate over non-zero entries of $f|_{B_k}$ and of $g|_{B'_k}$.
	Observe that each integer belongs to at most two blocks $B'_k$, so each non-zero entry of $g$ is considered for at most $2\sqrt{\delta \log \delta}$
	non-zero entries of $f$. Hence, the running time of this phase is also $\Oh(n\sqrt{\delta \log \delta})$.
\end{proof}

\fctdiff*
\begin{proof}
	Note that
	\[\Delta_p[f* g](i) = \sum_{j\in \mathbb{Z}}f(j)g(i+p-j)-\sum_{j\in \mathbb{Z}}f(j)\cdot g(i-j)=\sum_{j\in \mathbb{Z}}f(j)\cdot \Delta_p[g](i-j)=(f* \Delta_p[g])(i).\]
	By symmetry, we also have $\Delta_p[f*g] = \Delta_p[f]* g$.
	Consequently, $\Delta_p[f]*\Delta_p[g]= \Delta_p[\Delta_p[f]*g]=\Delta_p[\Delta_p[f* g]]$.
\end{proof}

\lemsimpler*
\begin{proof}
First, we assume that blocks of $d+p$ symbols $T[i,\ldots,i+d+p-1]$ can be processed simultaneously. We maintain the periodic representation of both $P$ and $T$ with respect to $p$.  Moreover, we store the values $(T\otimes P)(j)$ for $i-2p \le j < i$ (initialised as zeroes for $i=0$; this is valid due to \cref{fct:otimesham}).
The space consumption is $\Oh(d+p)$.

When the block arrives, we update the periodic representation of $T$ and apply \cref{cor:mpt} to compute $\Delta^2_p[T\otimes P](j)$ for $i\le j < i+d+p$.
Based on the stored values of $T\otimes P$, this lets us retrieve $(T\otimes P)(j)$ for $i\le j < i+d+p$.
Next, for each position $j>|P|-1$,  we report $\HS_{T,P}[j]=|P|-(T\otimes P)(j)$.
Finally, we discard the values $(T\otimes P)(j)$ for $j < i+d-p$.
Such an iteration takes $\Oh((d+p)\sqrt{(d+p)\log (d+p)})$ time and $\Oh(d+p)$ working space.

Below, we apply this procedure in a streaming algorithm which processes $T$ symbol by symbol.
The first step is to observe that if the pattern length is $\Oh(d+p)$, we can compute the Hamming distance online using $\Oh(d+p)$ words of space and $\Oh(\sqrt{(d+p) \log (d+p)})$ worst-case time per arriving symbol~\cite{CEPP:2008}. We now proceed under the assumption that $|P|>2(d+p)$.

We partition the pattern into two parts: the \emph{tail}, $P_T$~--- the suffix of $P$ of length $2(d+p)$, and the \emph{head}, $P_H$~--- the prefix of $P$ length $|P|-2(d+p)$. 
One can observe that $\Hs{P}{T}{j}=\Hs{P_T}{T}{j}+\Hs{P_H}{T}{j-2(d+p)}$. Moreover, we can compute $\Hs{P_T}{T}{j}$ using the aforementioned online algorithm of \cite{CEPP:2008}; this takes $\Oh(\sqrt{(d+p)\log (d+p)})$ time per symbol and $\Oh(d+p)$ words of space.

For the second summand, we need to ensure that we will have computed $\Hs{P_H}{T}{j-2(d+p)}$ before $T[j]$ arrives.
Hence, we partition the text into blocks of length $d+p$ and use our algorithm to process a block $T[i,\ldots,i+d+p-1]$ as soon as it is ready.
This procedure takes $\Oh((d+p)\sqrt{(d+p) \log (d+p)})$ time, so it is can be performed in the background while we read the next block $T[i+d+p,\ldots,i+2(d+p)-1]$.
Thus, $\Hs{P_H}{T}{j}$ is indeed ready on time for $j\in \{i,\ldots,i+d+p-1\}$.

The overall space usage is $\Oh(d+p)$ words and the worst-case time per arriving symbol is $\Oh(\sqrt{(d+p)\log (d+p)})$,
dominated by the online procedure of \cite{CEPP:2008} and by processing blocks in the background.
\end{proof}

\thmsmallperiod*
\begin{proof}
Our strategy is to partition $T$ into overlapping blocks for which $p$ is an $\Oh(k)$-period,
making sure that any $k$-mismatch occurrence of $P$ is fully contained within a block.
Then, we shall run \cref{lem:simpler} for each block to find these occurrences.

Let us first build the partition into blocks. We shall make sure that every position of $T$ belongs to exactly two blocks
and that $p$ is an $(4k+2d+p)$-period of each block. 
While processing $T$, we maintain two current blocks $B,B'$ (assume $|B|\ge |B'|$) along with their periodic representations.
Let us denote the number of mismatches (with respect to the approximate period $p$) in $B$ and $B'$ by $m$ and $m'$, respectively.
These values shall always satisfy $m' \le d+2k$ and $m  \le m'+\min(|B'|,p)+d+2k$,
which clearly guarantees the claimed bound $m \le 4k+2d+p$.
Moreover, we shall make sure that $d+2k < m$ unless $B$ is a prefix of $T$.
This way, if a $k$-mismatch occurrence of $P$ ends at the currently processed position, then it must be fully contained in $B$,
because $\HD(P,Q)\le k$ implies that $p$ is a $(d+2k)$-period of $Q$.

We start with $B=B'=\varepsilon$ before reading $T[0]$.
Next, suppose that we read a symbol $T[i]$.
If $m'<d+2k$ or $T[i]=T[i-p]$, we simply extend $B$ and $B'$ with $T[i]$.
In this case, $m'$ might increase but it will not exceed $d+2k$,
whereas $m$ increases only if $m'+\min(|B'|,p)$ increases, so the inequality $m \le m'+\min(|B'|,p)+d+2k$ remains satisfied.
On the other hand, if $m'=d+2k$ and $T[i]\ne T[i-p]$,
then we set $B := B'T[i]$ and $B' := T[i]$. 
In this case, $m = d+2k+1$ and $m' = 0$, which satisfies the invariants as one can easily verify.

The procedure described above outputs the blocks as streams, which we pass to \cref{lem:simpler} with a delay $\Delta$. 
In order to save some space, when the construction of a block terminates and the block turns out to be shorter than $|P|$,
we immediately launch a garbage collector to get rid of this block.
The number of remaining blocks contained in $T[i-\Delta+1,\ldots,i]$ is therefore bounded by $2\lfloor{\frac{\Delta}{|P|}}\rfloor$,
because each such block is of length at least $|P|$ and each position is located within at most two such blocks.
Accounting for the two blocks currently in construction and the two blocks currently processed by \cref{lem:simpler},
this implies that at any time we store $\Oh(1+\frac{\Delta}{|P|})=\Oh(1)$ blocks in total,
which means that the overall space consumption is bounded by $\Oh(k)$ words.

The instances of \cref{lem:simpler} report $k$-mismatch occurrences of $P$ with no delay, so the overall delay 
of the algorithm is precisely $\Delta$. Requests for mismatch information are handled in $\Oh(k)$ time using the periodic representations of 
the pattern $P$ and the currently processed block.

To allow for computing sketches, we also maintain an instance of \cref{cor:streamsketch} and for every block we $T[b,\ldots,e]$,
we store the sketch $\sk_k(T[0,\ldots,b-1])$. As we stream the block to \cref{lem:simpler},
we process it using another instance of \cref{cor:streamsketch}, with an extra delay $|P|$ (compared to \cref{lem:simpler}).
This way, whenever \cref{lem:simpler} reports a $k$-mismatch occurrence $T[j,\ldots,j+|P|-1]$,
the sketch $T[0,\ldots,j-1]$ can be retrieved (on demand) in $\Oh(k\log^2 n)$ time (by combining $\sk_k(T[0,\ldots,b-1])$ and $\sk_k(T[b,\ldots,j-1])$
with \cref{lem:manip}). The use of \cref{cor:streamsketch} increases the processing time of each position by an additive $\Oh(\log^2 n)$ term.
\end{proof}

\section{Proof of \cref{lem:key}}\label{sec:key_proof}

In this section, we prove \cref{lem:key}; its statement is repeated below for completeness. 
\keylem*

We define $\|\C_{p}(i)\|$ as the number of distinct elements in $\C_{p}(i)$;
note that a class is \emph{uniform} if $\|\C_{p}(i)\|=1$. 
The \emph{majority} element of a multiset $S$ is an element with multiplicity strictly greater than $\frac12 |S|$.
We define uniform strings and majority symbols of a string in an analogous way.

The remaining part of this section constitutes a proof of \cref{lem:key}.
We start with \cref{sec:overview}, where we introduce the main ideas, which rely on the structure of classes and their majorities.
The subsequent \cref{sec:comb} provides further combinatorial insight necessary to bound the size of our encoding.
\cref{sec:alg} presents two abstract building blocks based on well-known compact data structures.
Next, in \cref{sec:proof} we give a complete description of our encoding,
in \cref{sec:queries} we address answering queries, and in \cref{sec:updates} we discuss updates.

\subsection{Overview}\label{sec:overview}
Observe that if $d=\gcd(\PP)$ does not change as we insert an approximate period $p$ to $\PP$,
then we do not need to update the data structure. 
Hence, let us introduce a sequence $d_0,\ldots,d_s$ of \emph{distinct} values $\gcd(\PP)$ arising as we inserted subsequent approximate periods to $\PP$.
 Moreover, for $1\le i \le s$, let $p_i\in \PP$
be the period which caused the transition from $d_{i-1}$ to $d_i$.

\begin{fact}\label{fct:gcd}
The sequences $d_0,\ldots,d_s$ and $p_1,\ldots,p_s$ satisfy  $d_0 = 0$, 
$d_\ell = \gcd(d_{\ell-1},p_\ell)$ for $1\le \ell \le s$, and $d_s = \gcd(\PP)$.
Moreover, $d_{\ell} \mid  d_{\ell-1}$ and $d_{\ell} \le \frac n{2^{\ell+1}}$ for $1\le \ell \le s$,
and therefore $s = \Oh(\log n)$.
\end{fact}
\begin{proof}
We start with $\PP=\emptyset$, so $d_0 = \gcd\emptyset = 0$.
If $\gcd(\PP) \mid p$ for a newly inserted element $p$, we do not update the sequence.
Otherwise, we append $p_{\ell} := p$ and $d_{\ell} := \gcd(d_{\ell-1},p_{\ell})$.
Note that $d_{\ell}$ is a proper divisor of $d_{\ell-1}$, so $d_{\ell} \le \frac12 d_{\ell-1}$
which yields $d_{\ell}\le \frac{d_1}{2^{\ell-1}}\le \frac{n}{2^{\ell+1}}$ by induction.
\end{proof}

Let $\CCC_\ell$ be the partition of the symbols of $X$ into classes $\C_{d_{\ell}}(i)$ modulo $d_\ell$.
\Cref{fct:gcd} lets us characterise the sequence $\CCC_0,\ldots,\CCC_s$:
the first partition, $\CCC_0$, consists of singletons, i.e., it is the finest possible partition.
Then, each partition is coarser than the previous one, and finally $\CCC_s$ is the partition into classes modulo $d_s$. 

Consequently, the classes modulo $\C_{d_\ell}(i)$ for $0\le \ell \le s$ form a laminar family,
which can be represented as a forest of depth $s+1=\Oh(\log n)$; its leaves are single symbols (classes modulo $d_0=0$),
while the roots are classes modulo $d_s$.
Let us imagine that each class stores its majority element (or a sentinel $\#$ if there is no majority).
Observe that if all the classes $\C_{d_{\ell-1}}(i')$ contained in a given class $\C_{d_{\ell}}(i)$
share a common majority element, then this value is also the majority of $\C_{d_{\ell}}(i)$.
Consequently, storing the majority elements of all the contained classes $\C_{d_{\ell-1}}(i')$ is redundant.
Now, in order to retrieve $X[i]$, it suffices to start at the leaf $\C_{d_0}(i)$, walk up the tree until we reach a class
storing its majority, and return the majority, which is guaranteed to be equal to $X[i]$. 
This is basically the strategy of our query algorithm.
A minor difference is that we do not store the majority element of uniform classes $\C_{d_s}(i)$,
because our procedure shall return a sentinel $\#$ when $\C_{d_s}(i)$ is uniform.
On the other hand, we explicitly store the non-uniform classes $\C_{d_s}(i)$
so that updates can be implemented efficiently.

In order to encode the majority symbols of classes $\C_{d_{\ell-1}}(i')$ contained in a given class $\C_{d_{\ell}}(i)$,
let us study the structure of these classes in more detail.

\begin{observation}\label{obs:str}
Each class modulo $d_\ell$ can be decomposed as follows into non-empty classes modulo $d_{\ell-1}$;
\begin{align*}
\C_{d_\ell}(i) &= \bigcup_{j=0}^{\frac{d_{\ell-1}}{d_\ell}} \C_{d_{\ell-1}}(i+jp_\ell) &\text{if $\ell > 1$}\\
\C_{d_\ell}(i) &= \bigcup_{j=0}^{\lceil{\frac{n-i}{d_\ell}\rceil}} \C_{d_{\ell-1}}(i+jp_\ell) &\text{if $\ell = 1$}
\end{align*}
\end{observation}

Motivated by this decomposition, for each class $\C_{d_\ell}(i)$ with $\ell\ge1$ and $0\le i < d_\ell$,
we define the majority string $M_{\ell,i}$ of length $|M_{\ell,i}| = \frac{d_{\ell-1}}{d_\ell}$ for $\ell > 1$ and $|M_{\ell,i}|=\lceil{\frac{n-i}{d_\ell}\rceil}$ for $\ell=1$.
Its $j$-th symbol $M_{\ell,i}[j]$ is defined as the majority of $\C_{d_{\ell-1}}(i+jp_\ell)$,  or $\#$ if the class has no majority.
We think of $M_{\ell,i}$ as a cyclic string for $\ell>1$ and a linear string for $\ell=1$.

Since $p_\ell\in \Per(X,k)$, we expect that the adjacent symbols of the majority strings $M_{\ell,i}$ are almost always equal.
In the next section, we shall prove that the total number of mismatches between adjacent symbols is $\Oh(k)$ across all the majority strings.
\subsection{Combinatorial Bounds}\label{sec:comb}

For $1\le \ell \le s$, let $\NNN_\ell \sub \CCC_\ell$ consist of non-uniform classes.
Moreover, for $0\le \ell < s$, let $\KKK_\ell \sub \CCC_\ell$ consist of classes $\C_{d_\ell}(i)$ such that the majority elements of  $\C_{d_\ell}(i)$ and $\C_{d_\ell}(i+p_{\ell+1})$ differ.

\begin{fact}\label{fct:nnk}
Consider the decomposition of a class $\C_{d_\ell}(i)\in \NNN_\ell$ into classes $C\in \CCC_{\ell-1}$.
At least one of these classes satisfies $C\in \NNN_{\ell-1}\cup \KKK_{\ell-1}$.
Moreover, if there is just one such class, then $\ell=1$ or this class $C$ satisfies $C \in \NNN_{\ell-1}\sm \KKK_{\ell-1}$.
\end{fact}
\begin{proof}
If the majority string $M_{\ell,i}$ is uniform, then one of the classes $C$
must contain a symbol other than its majority; otherwise, $\C_{d_\ell}(i)$ would be uniform.
Such a class $C$ clearly belongs to $\NNN_{\ell-1}\sm \KKK_{\ell-1}$.

Next, suppose that the majority string $M_{\ell,i}$ is non-uniform.
Each mismatch between consecutive symbols of $M_{\ell,i}$ corresponds to a class $C \in \KKK_{\ell-1}$.
If $\ell=1$, then $M_{1,i}$ is a linear string and it may have one mismatch.
Otherwise, $M_{\ell,i}$ is circular, so there are at least two mismatches between consecutive symbols.
\end{proof}

\begin{fact}\label{fct:ksmn}
If $\C_{d_\ell}(i)\in \KKK_{\ell}\sm \NNN_{\ell}$ for some $0\le \ell< s$, then there are at least $2^{\ell-1}$ positions $i'\equiv i \pmod{d_{\ell}}$ such that $0 \le i' < n-p_{\ell+1}$ and $X[i']\ne X[i'+p_{\ell+1}]$.
\end{fact}
\begin{proof}
If $\ell=0$, then we just have $\C_{0}(i)\in \KKK_0$ if and only if $0\le i < n-p_1$ and $X[i]\ne X[i+p_1]$.

Consider an alignment between $X[0,\ldots,n-p_{\ell+1}-1]$ and $X[p_{\ell+1},\ldots,n-1]$
and let $k_{\ell,i}$ be the number of positions $i'$ specified above.
Observe that exactly $\lfloor\frac{i+p_{\ell+1}}{d_\ell}\rfloor$ symbols in $\C_{d_\ell}(i+p_{\ell+1})$ are at indices smaller than $p_{\ell+1}$ 
(and they are not aligned with any symbol of $\C_{d_\ell}(i)$), while exactly $k_{\ell,i}$ symbols are aligned with mismatching symbols.
The remaining symbols of $\C_{d_\ell}(i+p_{\ell+1})$ are aligned with matching symbols of $\C_{d_\ell}(i)$.
The class $\C_{d_\ell}(i)$ is uniform, so at most $k_{\ell,i} + \lfloor\frac{i+p_{\ell+1}}{d_\ell}\rfloor$ symbols of $\C_{d_\ell}(i+p_{\ell+1})$ are not equal to the majority of $\C_{d_\ell}(i)$. 
Since $\C_{d_\ell}(i+p_{\ell+1})$ does not share the majority with $\C_{d_\ell}(i)$,
we must have $k_{\ell,i} + \lfloor\frac{i+p_{\ell+1}}{d_\ell}\rfloor \ge \tfrac{1}{2}|\C_{d_\ell}(i+p_{\ell+1})|=\frac12\lfloor{\frac{i+p_{\ell+1}}{d_\ell}}\rfloor +\frac12\lceil\frac{n-i-p_{\ell+1}}{d_\ell}\rceil.$
Due to $p_{\ell+1}\le \frac{n}{4}$ and $d_\ell \le \frac{n}{2^{\ell+1}}$ (by \cref{fct:gcd}), this yields 
\[k_{\ell,i} \ge \tfrac{1}{2}\ceil{\tfrac{n-i-p_{\ell+1}}{d_\ell}} -\tfrac12\floor{\tfrac{i+p_{\ell+1}}{d_\ell}}\ge \tfrac{n-2(i+p_{\ell+1})}{2d_\ell}=
\tfrac{2n-4i-4p_{\ell+1}}{4d_\ell} > \tfrac{2n-4d_{\ell}-n}{4d_\ell}=\tfrac{n}{4d_\ell}-1\ge 2^{\ell-1}-1.\]
In short, $k_{\ell,i} > 2^{\ell-1}-1$, and thus $k_{\ell,i}\ge 2^{\ell-1}$. 
\end{proof}

\begin{lemma}\label{lem:maj}
We have
 $2|\NNN_{s}| + \sum_{\ell=0}^{s-1}|\KKK_\ell| \le 8k$.
Consequently, the majority strings $M_{\ell,i}$ contain in total at most $8k$ mismatches between adjacent symbols
and $\sum_{C\in \NNN_{d_s}} \|C\| \le 16k$.
\end{lemma}
\begin{proof}
We apply a discharging argument. In the charging phase,
each mismatch $X[i]\ne X[i+p_\ell]$ (for $1\le \ell \le s$ and $0 \le i < n-p_\ell$) receives a charge of $2^{3-\ell}$ units.
The total charge is therefore at most $\sum_{\ell=1}^{s} k\cdot 2^{3-\ell} < 8k$.

Next, each such mismatch passes its charge to the class $\C_{d_\ell}(i)$.
By \cref{fct:ksmn}, each class $\C_{d_\ell}(i)\in \KKK_\ell \sm \NNN_\ell$ receives at least 2 units of charge.
Moreover, each class $\C_{d_{0}}(i) \in \KKK_0 \sm \NNN_0$ receives exactly 4 units.

Finally, in subsequent iterations for $\ell=0$ to $s-1$, the classes modulo $d_{\ell}$ pass some charge to the enclosing classes modulo $d_{\ell+1}$:
each $\C_{d_\ell}(i)\notin \KKK_\ell$ passes all its charge to $\C_{d_{\ell+1}}(i)$,
whereas each $\C_{d_\ell}(i)\in \KKK_\ell$ leaves one unit for itself and passes the remaining charge.

We shall inductively prove that prior to the iteration $\ell$, each class $\C_{d_\ell}(i)\in \NNN_\ell$
had at least two units of charge. Let us fix such a class.
If $\ell=1$, then \cref{fct:nnk} implies that it contains a class $C\in \KKK_0$ (as $\NNN_0 = \emptyset$).
As we have observed, it obtained 4 units of charge and passed 3 of them to $\C_{d_\ell}(i)$.
Similarly, if $\C_{d_\ell}(i)$ contains a class $C\in \NNN_{\ell-1}\sm \KKK_{\ell-1}$,
then this class obtained at least 2 units of charge (by the inductive assumption), and passed them all to $\C_{d_\ell}(i)$.
Otherwise, \cref{fct:nnk} tells us that $\C_{d_\ell}(i)$ contains at least two classes $C\in \NNN_{\ell-1}\cup \KKK_{\ell-1}$.
Each of them received at least 2 units of charge (directly from the mismatches or due to the inductive assumption)
and passed at least 1 unit to $\C_{d_\ell}(i)$. 

In the end, each class $C\in \KKK_{\ell}$ has therefore at least one unit of charge
and each class $C\in \NNN_{s}$ has at least 2 units.
This completes the proof of the inequality.

For the remaining two claims, observe that mismatches in the majority strings correspond to classes $C\in \KKK_\ell$,
and that if $a\in C$ for $C\in \NNN_s$, then there exists $\C_{d_\ell}(i)\sub \C$
such that the majority string $M_{\ell,i}$ is non-uniform and contains $a$. Consequently, $\|C\|$
is bounded by twice the number of mismatches in the corresponding majority strings.
The classes modulo $d_s$ are disjoint, so no mismatch is counted twice.
\end{proof}

\subsection{Algorithmic Tools}\label{sec:alg}
A maximal uniform fragment of a string $S$ is called a \emph{run}; we denote the number of runs  by $\rle(S)$.
\begin{fact}[Run-length encoding]\label{fct:rle}
A string $S$ of length $n$ with $r = \rle(S)$ can be encoded using $\Oh(r(\log \frac{n+r}{r} + \log |\Sigma|))$ bits
so that any given symbol $S[i]$ can be retrieved in $\Oh(\log r)$ time.
This representation can be constructed in $\Oh(r\log n)$ time from the run-length encoding of $S$.
\end{fact} 
\begin{proof}
Let $0= x_1 < \ldots <x_r < n$ be starting position of each run.  
We store the sequence $x_1,\ldots,x_r$ using the Elias--Fano representation~\cite{DBLP:journals/jacm/Elias74,Fano} (with $\Oh(1)$-time data structure for selection queries in a bitmask; see e.g.~\cite{CM:1996}). 
It takes $\Oh(r\log \frac{n+r}{r} + r)$ bits and allows $\Oh(1)$-time access.
In particular, in $\Oh(\log r)$ time we can binary search for the run containing a given position $i$.
The values $X[x_1],\ldots,X[x_r]$ are stored using $\Oh(r \log |\Sigma|)$ bits with $\Oh(1)$-time access.
\end{proof}

\begin{fact}[Membership queries, \cite{DBLP:journals/siamcomp/BrodnikM99}]\label{fct:memb}
A set $A \sub \{0,\ldots,n-1\}$ of size at most $m$ can be encoded in $\Oh(m\log\frac{n+m}{m})$ bits so
that one can check in $\Oh(1)$ time whether $i\in A$ for a given $i \in \{0,\ldots,n-1\}$.
The construction time is $\Oh(m\log n)$.
\end{fact}

\newcommand{\M}{\mathbf{M}}

\subsection{Data Structure}\label{sec:proof}

Following the intuitive description in \cref{sec:overview}, we shall store all the non-uniform majority strings $M_{\ell,i}$ (for $1\le \ell \le s$ and $0 \le i < d_\ell$)
and all non-uniform classes $\C_{d_s}(i)$. 
We represent them as non-overlapping factors of a single string $\M$ of length $2n$, constructed as follows:
Initially, $\M$ consists of blank symbols $\diamond$.
Each non-uniform majority string $M_{\ell,i}$ is placed in $\M$ at position 
$2d_{\ell} + i\frac{d_{\ell-1}}{d_\ell}$ for $\ell > 1$ and $2d_{\ell} + i\lceil{\frac{n}{d_\ell}}\rceil$ for $\ell = 1$.
Note that the positions occupied by strings $M_{\ell,i}$ for a fixed level $\ell$
belong to the range  $[2d_\ell, 2d_\ell + d_{\ell-1}-1]\sub [2d_{\ell}, 2d_{\ell-1}-1]$
for $\ell > 1$ and $[2d_1,\ldots, 2d_1 + d_1 \lceil{\frac{n}{d_\ell}}\rceil-1]\sub [2d_1,3d_1+n-1]\sub [2d_1,2n-1]$ for $\ell=1$.
These ranges are clearly disjoint for distinct values $\ell$, so the majority strings $M_{\ell,i}$ indeed do not overlap.
Additionally, we exploit the fact that positions within $[0,\ldots,d_{s}-1]$ are free,
and if $\C_{d_s}(i)\in \NNN_s$, we store its majority symbol at $\M[i]$. 
\Cref{lem:maj} yields that the total number of mismatches between subsequent symbols 
and the number of non-uniform classes modulo $d_s$ are both $\Oh(k)$.
Hence, $\rle(\M)= \Oh(k)$ and the space required to store $\M$ using \cref{fct:rle} is $\Oh(k(\log \frac{n+k}{k} + \log |\Sigma|))$ bits.
On top of that, we also store a data structure of \cref{fct:memb} marking the 
positions in $\M$ where non-uniform majority strings start; this component takes $\Oh(k\log \frac{n+k}{k})$ bits.

Additionally, we keep the contents of each non-uniform class modulo $d_s$. We do not need to access this data efficiently,
so for each such class, we simply store the symbols and their multiplicities using variable-length encoding.
This takes $\Oh(\|\C\|(\log |\Sigma|+\log \frac{|C|+\|C\|}{\|C\|}))$ bits for each $C\in \NNN_s$, which is $\Oh(k(\log \frac{n}{k}+\log |\Sigma|))$ 
in total because $\sum_{C\in \NNN_{s}}\|C\|=\Oh(k)$ (by \cref{lem:maj}) and $\sum_{C\in \NNN_{s}}|C|\le n$.

Finally, we store integers $n$, $d_1$, $d_s$, as well as $\frac{d_{\ell}}{d_{\ell+1}}$ and $r_\ell :=(\frac{p_{\ell+1}}{d_{\ell+1}})^{-1} \mod \frac{d_{\ell}}{d_{\ell+1}}$ for $1\le \ell < s$.
A naive estimation of the required space is $\Oh(s \log  n) = \Oh(\log^2  n)$ bits, but variable-length encoding
lets us store the values $\frac{d_{\ell}}{d_{\ell+1}}$ using $\Oh(\sum_{\ell=1}^s \log \frac{d_{\ell}}{d_{\ell+1}}) = \Oh(\log n)=\Oh(k\log \frac{n+k}{k})$ bits in total.
Similarly, the integers $r_{\ell}$ can be stored in $\Oh(\log n)$ bits because $1 \le r_\ell < \frac{d_{\ell}}{d_{\ell+1}}$.

This completes the description of our data structure;
its takes $\Oh(k(\log\frac{n+k}{k} +\log |\Sigma|))$ bits.

\subsection{Queries}\label{sec:queries}
In this section, we describe the query algorithm for a given index $i$.
We are going to iterate for $\ell=0$ to $s$, and for each $\ell$ we will either learn $X[i]$
or find out that $X[i]$ is the majority symbol of $\C_{d_{\ell}}(i)$. Consequently, entering iteration $\ell$, we already know that $X[i]$ is the majority of $\C_{d_{\ell-1}}(i)$.
We also assume that $d_\ell$ is available at that time.

We compute the starting position of $M_{\ell,i \bmod d_\ell}$ in $\M$ according to the formulae given in \cref{sec:proof}.
Next, we query the data structure of \cref{fct:memb} to find out if the majority string is uniform.
If so, we conclude that $X[i]$ is the majority of $\C_{d_{\ell}}(i)$ and we may proceed to the next level.
Before this, we need to compute $d_{\ell+1} = d_\ell \cdot (\frac{d_\ell}{d_{\ell+1}})^{-1}$.
An iteration takes $\Oh(1)$ time in this case.

Otherwise, we need to learn the majority of $\C_{d_{\ell-1}}(i)$, which is guaranteed to be equal to $X[i]$.
This value is $M_{\ell,i \bmod d_\ell}[j]$ where $j = \lfloor\frac{i}{d_1}\rfloor$ for $\ell=1$,
and $j=r_\ell \lfloor{\frac{i}{d_{\ell}}\rfloor} \bmod \frac{d_{\ell-1}}{d_{\ell}}$ for $\ell > 1$.
We know the starting position of $M_{\ell,i \bmod d_\ell}$ in $\M$, so we just use \cref{fct:rle}
to retrieve $X[i]$ in $\Oh(\log k)$ time.

If the query algorithm completes all the $s$ iterations without exiting,
then $X[i]$ is guaranteed to be the majority symbol of $\C_{d_s}(i)$.
Thus, we retrieve $\M[i \bmod d_s]$, which takes $\Oh(\log k)$ time due to \cref{fct:rle}.
This symbol is either the majority of $\C_{d_s}(i)$ (guaranteed to be equal to $X[i]$)
or a blank symbol. 
In the latter case, we know that the class $\C_{d_s}(i)$ is uniform,
so we return a sentinel $\#$.
The overall running time is $\Oh(s + \log k) = \Oh(\log n)$.

\subsection{Updates}\label{sec:updates}
If $d_s \mid p$, then the we do not need to update.
Otherwise, we extend our data structure with $p_{s+1}:=p$ and $d_{s+1} := \gcd(p_{s+1},d_s)$.

First, we shall detect non-uniform classes modulo $d_{s+1}$.
\cref{fct:nnk,fct:ksmn} imply that a class $\C_{d_{s+1}}(i)$ is non-uniform if and only if 
for some $i'\equiv i \pmod{d_{s+1}}$ there is a class $\C_{d_s}(i')\in \NNN_{s}$ or a position $0\le i' < n-p_{s+1}$
with $X[i']\ne X[i'+p_{s+1}]$. Hence, we scan the non-uniform classes modulo $d_s$ and the mismatch information for $p_{s+1}$
grouping the entries by $i' \bmod d_{s+1}$.

For each $\C_{d_{s+1}}(i)\in \NNN_{s+1}$ our goal is to construct the underlying multiset
and the corresponding majority string $M_{s+1,i}$.
For every enclosed $C\in \NNN_s$, we use the underlying multiset to deduce the majority symbol
and store it at an appropriate position of $M_{s+1,i}$.
Next, for every mismatch $X[i']\ne X[i'+p_{s+1}]$
we store $X[i']$ in $M_{s+1,i}$ as the majority symbol of $\C_{d_s}(i')$ if the class is uniform.
Symmetrically, if $\C_{d_s}(i'+p_{s+1})$ is uniform, its majority symbol $X[i'+p_{s+1}]$ is placed in $M_{s+1,i}$.
Now, the remaining symbols of $M_{s+1,i}$ are guaranteed to be equal to their both neighbours.
This lets us retrieve the run-length encoding of $M_{s+1,i}$.
To compute the multiset of $\C_{d_{s+1}}(i)$, we aggregate the data from the enclosed non-uniform classes,
and use the majority string to retrieve for each symbol $a$ the total size of enclosed classes uniform in $a$.

Finally, we note that the data structures of \cref{fct:rle,fct:memb}
can be reconstructed in $\Oh(k\log n)$ time (which is sufficient to build them from scratch).

\section{Omitted Proofs from \cref{sec:det}}
\fctretr*
	\begin{proof}
	First, note that $\C_d(P,i)\cup \C_d(T',i)$ is uniform if and only if $\C_d(P,i)$
	is uniform and the mismatch information for neither stored $k$-mismatch occurrence contains $(i',P[i'], T[i''])$
	with $i'\equiv i \pmod{d}$. Hence, Bob can easily decide which positions of $P_\#$ and $T'_\#$ contains sentinel symbols.
	
	Next, we shall prove that Bob can retrieve $P_\#[i]=P[i]$ if $\C_d(P,i)\cup \C_d(T',i)$ is non-uniform.
	If $\C_d(P,i)$ is non-uniform, we can simply use the data structure of \cref{lem:key}.
	Otherwise, mismatch information for the $k$-mismatch occurrence of $P$ as a prefix or as a suffix of $T'$ contains $(i',P[i'],T[i''])$ for some $i'\equiv i \pmod{d}$. 
	The class $\C_d(P,i)$ is uniform, so Bob learns $P[i]=P[i']$.
	
	Finally, consider retrieving $T'_\#[i]=T'[i]$.
	If $i < n$, then Bob can use $P[i]$ and the mismatch information of the $k$-mismatch occurrence of $P$ as a prefix of $T'$,
	which might contain $(i, P[i], T'[i])$.
	On the other hand, for $i'\ge n$ he can use $P[i-\ell'+\ell]$ and the mismatch information of the $k$-mismatch occurrence of $P$ as a suffix of $T'$.
	\end{proof}

\fctequiv*
	\begin{proof}
	Suppose that $P$ has a $k$-mismatch occurrence at position $j$ of $T$.
	By definition of $\PP$, we conclude that $\ell \le j \le \ell'$ and $j \equiv \ell \pmod{d}$.
	Consequently, $P$ has a $k$-mismatch occurrence at position $j':=j-\ell$ of $T'$ satisfying $d \mid j'$.
	If $P[i]\ne T'[j'+i]$, then the class $\C_d(P,i)\cup \C_d(T',i)$ contains at least two distinct elements,
	so $P_\#[i] = P[i]\ne T'[j'+i]=T'_\#[j'+i]$.
	Otherwise, $P_\#[i]=T'_\#[j+i]$, with both symbols equal to $\#_{i\bmod d}$ or $P[i]=T'[i+j]$.
	Thus, $\MM(P,T[j,\ldots,j+n-1])=\MM(P_\#,T'_\#[j-\ell,\ldots,j-\ell+n-1])$ and, in particular,
	$P_\#$ has a $k$-mismatch occurrence at position $j'$ of $T'_\#$.
	
	For the converse proof, suppose $P_\#$ has a $k$-mismatch occurrence at position $j'$ of $T'_\#$,
	with $j'=j-\ell$. If $P_\#[i]=T'_\#[j'+i]$, then either $P[i]=P_{\#}[i]=T'_\#[j'+i]=T'[j'+i]$,
	or $P_{\#}[i]=T'_\#[j'+i]=\#_{i\bmod p}$. In the latter case, we conclude that $d \mid j'$
	and  $\C_d(P,i)\cup \C_d(T',i)$ is uniform, so $P[i]=T'[j'+i]$.
	Hence, $P$ indeed has a $k$-mismatch occurrence at position $j'$ of $T'$,
	i.e., at position $j$ of $T$.
	\end{proof}

\section{Omitted Proofs from \cref{sec:rand}}
In order to prove \cref{thm:buffer}, we first give a new efficient algorithm for the streaming $k$-mismatch problem,
assuming we can maintain a read-only copy of the latest $n$ symbols in the text.  The algorithm we give matches the running time of the classic offline algorithm of~\cite{ALP:2004} despite having guaranteed worst-case performance per arriving symbol and using small additional space on top of that, needed to store the last $n$ symbols of the text. 

\thmram*
\begin{proof}
First, we briefly sketch our strategy.
If the pattern has an $\Oh(k)$-period $\Oh(k)$, then it suffices to apply \cref{lem:perpref,thm:small-period}.
Otherwise, we can still use these results to filter the set of positions where a $P$ has $k$-mismatch occurrence in $T$,
leaving at most one candidate for each $k$ subsequent positions. We use sketches to verify candidates,
with the \emph{tail trick} (see~\cref{lem:simpler}) employed to avoid reporting occurrences with a delay.

More formally, while processing the pattern, we also construct a decomposition $P=P_HP_T$ into the \emph{head} $P_H$
and the \emph{tail} $P_T$ with $|P_T|=2k$, and we compute the sketch $\sk_k(P_H)$ (using \cref{cor:streamsketch}).
We also use \cref{lem:perpref} with $p=k$ and $d=2k+1$, which results in a prefix $Q$ of $P$.
This way, $P$ is processed in $\Oh(\sqrt{k\log k}+\log^2 n)$ time per symbol using $\Oh(k\log n)$ bits of working space.

If $|Q|> |P_H|$, then $P$ has a $4k$-period $p'\le k$, and we may just use \cref{thm:small-period} to report the $k$-mismatch occurrences of $P$ in $T$.
Otherwise, $Q$ has a $(2k+1)$-period $p'\le k$, but no $2k$-period $p''\le k$.
In particular, due to \cref{fct:per}, the $k$-mismatch occurrences of $Q$ are located more than $k$ positions apart.

Processing the text $T$, we apply \cref{cor:streamsketch} so that $\sk_k(T[0,\ldots,i])$ and $\sk_k(T[0,\ldots,i-|P_H|])$
can be efficiently retrieved while $T[i]$ is revealed. Additionally, we run the streaming algorithm of \cref{thm:small-period},
delayed so that a $k$-mismatch occurrence of $Q$ starting at position $i-|P_H|+1$ is reported while $T[i]$ is revealed.
These components take $\Oh(k\log n)$ bits of space and use $\Oh(\sqrt{k\log k}+\log n)$ time per symbol of $T$.

If \cref{thm:small-period} reports a $k$-mismatch occurrence of $Q$ at positions $i-|P_H|+1$,
we shall check if $P_H$ also has a $k$-mismatch occurrence there.
For this, we retrieve $\sk_k(T[0,\ldots,i])$ and $\sk_k(T[0,\ldots,i-|P_H|])$ (using \cref{cor:streamsketch}),
compute $\sk_k(T[i-|P_H|+1,\ldots,i])$ (using \cref{lem:manip}), and compare it to $\sk_k(P_H)$ (using \cref{cor:sketches}).
This process $\Oh(k \log^3 n)$ time in total, and it can be run while the subsequent $k$ symbols of $T$ are revealed.
If $P_H$ has a $k$-mismatch occurrence at position $i-|P_H|+1$, it results in $\MM(P_H, T[i-|P_H|+1,\ldots,i])$.
Then, we naively compute $\MM(P_T, T[i+1,\ldots,i+2k])$.
Thus, as soon as $T[i+2k]$ is revealed, we know if $\HD(P, T[i-|P_H|+1,\ldots,i+2k])\le k$,
and we can retrieve the mismatch information in $\Oh(k)$ time upon request.

Since $Q$ does not have any $2k$-period $p''\le k$, we are guaranteed that at most two $k$-mismatch occurrences of $Q$
are processed in parallel.
\end{proof}

We are now able to prove \cref{thm:buffer}.

\thmbuffer*
\begin{proof}
We partition $T$ into consecutive blocks of length $b=\frac14\min(\Delta,|P|)$.
The buffer shall be implemented as an \emph{assembly line} of components, each responsible 
for $k$-mismatch occurrences of $P$ starting within a single block, called the \emph{relevant occurrences} in what follows.

The choice of $b$ guarantees that storing $\Oh(1)$ components suffices at any time.
Moreover, the component needs to output the $k$-mismatch occurrences 
$\Theta(|P|)$ ticks after it is fed with the last relevant $k$-mismatch occurrence,
which leaves plenty of time for reorganisation.
Consequently, its lifetime shall consist of three phases:
\begin{itemize}[itemsep=0pt,topsep=5pt,parsep=2pt]
\item \emph{compression}, when it is fed with relevant $k$-mismatch occurrences of $P$ in $T$,
\item \emph{reorganisation}, when it performs some computations to change its structure,
\item \emph{decompression}, when it retrieves the stream of $k$-mismatch information of $P$ in $T$.
\end{itemize}

In the compression phase, we essentially construct the message as described in the proof of \cref{thm:prob-one},
encoding the relevant occurrences of $P$, i.e., the $k$-mismatch occurrences of $P$ in the appropriate fragment of $T$.
The only difference is that we also store the sketches $\sk_k(T[0,\ldots,\ell-1])$
and $\sk_k(T[0,\ldots,\ell'-1])$ corresponding to the leftmost and the rightmost relevant occurrence.

Processing a single relevant occurrence is easily implementable in $\Oh(k \log n)$ time,
dominated by updating the data structure of \cref{lem:key}, which may need to account for a new element of $\PP$. 
Apart from that, we only need to replace the rightmost relevant occurrence and the associated information.

In the decompression scheme, we apply \cref{thm:ram} to find $k$-mismatch occurrences of $P_\#$
in $T'_\#$, as defined in the proof of \cref{thm:prob-one}.
To provide random access to these strings,
we just need the data structure of \cref{lem:key} and the mismatch information for the two extremal relevant occurrences of $P$.
To allow for $\Oh(\log n)$-time access, we organise the mismatch information in two dictionaries:
for each mismatch $(i,P[i],T[i'])$, we store $T[i']$ in a dictionary indexed by $i'$,
and $P[i]$ in a dictionary indexed by $i \bmod d$. 
These dictionaries can be constructed in $\Oh(k)$ time (during the reorganisation phase).
As a result, we can use \cref{thm:ram} to report the occurrences of $P$ in $T$ in the claimed running time, along with the mismatch information.
The reorganisation phase is also used to process the first $n$ symbols of $T'_\#$.

Retrieving  the corresponding sketches $T[0,\ldots,i-1]$ is more involved, and here is where the reorganisation phase becomes useful again.
Based on the sketches $\sk_k(T[0,\ldots,\ell-1])$ and $\sk_k(T[0,\ldots,\ell'-1])$,
we can compute $\sk_k(T[\ell,\ldots,\ell'-1])=\sk_k(T'[0,\ldots,\ell'-\ell-1])$ applying \cref{lem:manip}.
Consider the point-wise difference $D$ of strings $T'$ and $T'_\#$.
Observe that \cref{cor:streamsketch} lets us transform $\sk_k(T'[0,\ldots,\ell'-\ell-1])$ into $\sk_k(D[0,\ldots,\ell'-\ell-1])$
using random access to $T'_\#$ for listing mismatches.
Next, we observe that $D$ is an integer power of a string of length $d$, so \cref{lem:manip} can be used to retrieve the sketch $\sk_k(D[0,\ldots,d-1])$
of its root.

During the decompression phase, we maintain the data structure of \cref{cor:streamsketch} transforming $\sk_k(D[0,\ldots,\ell'-\ell-1])$
back to $\sk_k(T'[0,\ldots,\ell'-\ell-1])$. 
When a $k$-mismatch occurrence of $P_\#$ is reported at position $i$ of $T'_\#$, we report a $k$-mismatch occurrence of $P$ at position $\ell+i$.
At the same time, we compute $\sk_k(D[i,\ldots,\ell'-\ell-1])$ (using \cref{lem:manip}; we are guaranteed that $d \mid i$) and retrieve $\sk_k(T'[0,\ldots,i-1]D[i,\ldots,\ell'-\ell-1])$
(from \cref{cor:streamsketch}). Finally, we construct $\sk_k(T[\ell,\ldots,\ell+i-1])$ and $\sk_k(T[0,\ldots,\ell+i-1])$ using \cref{lem:manip}.

In the compression phase, we need to update anything only when the component is fed with a relevant $k$-mismatch occurrence,
and processing such an occurrence takes $\Oh(k \log n)$ time.
The reorganisation time is $\Oh(n (\sqrt{k \log k}+ \log^3 n))$, which is $\Oh(\sqrt{k \log k}+ \log^3 n)$ time per tick.
In the decompression phase, the running time is $\Oh(\sqrt{k \log k}+ \log^3 n)$ per tick (due to \cref{thm:ram,cor:streamsketch})
plus $\Oh(k\log^2 n)$ time whenever a $k$-mismatch occurrence of $P$ is reported.
\end{proof}

\bibliographystyle{plainurl}

\end{document}